\documentclass[11pt]{amsart}
\usepackage{amsmath, amssymb, amsbsy, amsfonts, amsthm, latexsym, amsopn, amstext, amsxtra, euscript, amscd, color, mathrsfs}
\usepackage[normalem]{ulem}
\usepackage{soul}
\numberwithin{table}{section}
\PassOptionsToPackage{hyphens}{url}\usepackage{hyperref}
\makeatletter
\@namedef{subjclassname@2020}{%
  \textup{2020} Mathematics Subject Classification}
\makeatother
 
 \usepackage[capbesideposition=outside,capbesidesep=quad]{floatrow}
\usepackage{longtable}
\usepackage{tabularx}

\restylefloat{table}
\restylefloat{table}
      
\usepackage{multirow,caption}
\usepackage{array}
\usepackage{amscd}
\usepackage{color,enumerate}
\newcommand{\RNum}[1]{\lowercase\expandafter{\romannumeral #1\relax}}

\usepackage[colorinlistoftodos,prependcaption,textsize=tiny]{todonotes}

\newtheorem{thm}{Theorem}[section]
\newtheorem{lem}[thm]{Lemma}
\newtheorem{cor}[thm]{Corollary}
\newtheorem{prop}[thm]{Proposition}

\newtheorem{rmk}[thm]{Remark}

\newtheorem{thm-con}[thm]{Theorem-Conjecture}
\numberwithin{equation}{section}

\theoremstyle{definition}
\newtheorem{defn}[thm]{Definition}

\newcommand{\cB}{\mathcal B}

\newcommand{\F}{\mathbb F}

\newcommand{\VB}{\mathcal{VB}}
\def\Tr{{\rm Tr}}
\def\Trn{{\rm Tr}_1^n}

\DeclareMathOperator{\FB}{FBCT}
\begin{document}
\title[The second-order zero differential spectra of some functions]{The second-order zero differential spectra of some functions over finite fields}
 
  \author[K. Garg]{Kirpa Garg}
  \address{Department of Mathematics, Indian Institute of Technology Jammu, Jammu 181221, India}
  \email{kirpa.garg@gmail.com}
  
  \author[S. U. Hasan]{Sartaj Ul Hasan}
  \address{Department of Mathematics, Indian Institute of Technology Jammu, Jammu 181221, India}
  \email{sartaj.hasan@iitjammu.ac.in}
  
  \author[C. Riera]{Constanza Riera}
  \address{Department of Computer Science, Electrical Engineering and Mathematical Sciences, Western Norway University of Applied Sciences, 5020 Bergen, Norway}
  \email{csr@hvl.no}
  
   \author[P.~St\u anic\u a]{Pantelimon~St\u anic\u a$^1$}
   \address{Applied Mathematics Department, Naval Postgraduate School, Monterey, CA 93943, USA}
  \email{pstanica@nps.edu}

\thanks{The work of K. Garg is supported by the University Grants Commission (UGC), Government of India. 
The work of S. U. Hasan is partially supported by Core Research Grant CRG/2022/005418 from the Science and Engineering Research Board, Government of India. The work of P. St\u anic\u a   is partially supported by a grant from the NPS Foundation.}
\thanks{$^1$Corresponding author.}



\begin{abstract}  It was shown by Boukerrou et al.~[IACR Trans. Symmetric Cryptol. 1 (2020), 331--362]  that the $F$-boomerang uniformity (which is the same as the second-order zero differential uniformity in even characteristic) of perfect nonlinear functions is~$0$ on $\F_{p^n}$ ($p$ prime) and the one of almost perfect nonlinear functions on $\F_{2^n}$ is~$0$. It is natural to inquire what happens with APN or other low differential uniform  functions in even and odd characteristics. 
Here, we explicitly determine the second-order zero differential spectra of several maps with low differential uniformity. In particular, we compute the second-order zero differential spectra for some almost perfect nonlinear  (APN) functions  over finite fields of odd characteristic, pushing further the study started in Boukerrou et al. and continued in Li et al.~[Cryptogr. Commun. 14.3 (2022), 653--662], and it turns out that our considered functions also have low second-order zero differential uniformity.    Moreover, we study the second-order zero differential spectra of certain functions with low differential uniformity over finite fields of even characteristic. We connect this new concept to the sum-freedom and vanishing flats concepts and find some counts for the number of vanishing flats via our methods.
We provide detailed analyses on several equations over finite fields that may have an interest outside of the scope of our paper.
\end{abstract}
\maketitle

\section{Introduction}
Let $n$ be a positive integer and $p$ be a prime number. We denote by $\F_q$ the finite field with $q=p^n$ elements, by $\F_{q}^{*}$ the multiplicative cyclic group of nonzero elements of $\F_{q}$ and by $\F_{q}[X]$ the ring of polynomials in one variable $X$ with coefficients in $\F_{q}$. It may be noted that functions over finite fields are very important objects due to their wide range of applications in coding theory and cryptography. For example, in cryptography, these functions are often used in designing what are known as substitution boxes (S-boxes) in modern block ciphers  (mostly, for $p=2$, though there are some proposals in odd characteristic such as GMiMC~\cite{GMiMC}). One of the most effective attacks on block ciphers is differential cryptanalysis, which was first introduced by Biham and Shamir~\cite{Biham91}. The resistance of a function against differential attack is measured in terms of its differential uniformity (DU) -- a notion introduced by Nyberg~\cite{Nyberg}.  For a function $F : \F_{q}\to \F_{q}$, and any $a \in \F_{q}$, the derivative of $F$ in the direction of $a$ is defined as $D_F(X,a) := F(X+a)-F(X)$ for all $X\in \F_{q}$. For any $a, b \in \F_q$, the Difference Distribution Table (DDT) entry $\delta_F(a,b)$ at point $(a, b)$ is the number of solutions $X\in \F_{q}$ of the equation $D_F(X,a) = b$. Further, the differential uniformity of $F$, denoted by $\delta_F$, is given by $\delta_{F} := \max\{\delta_{F}(a, b) : a\in \F_{q}^*, b \in \F_{q} \}.$ We call the function $F$ a perfect nonlinear (PN) function, respectively, an almost perfect nonlinear (APN) function, if $\delta_{F} = 1$, respectively, $\delta_{F} = 2$. Blondeau, Canteaut, and Charpin~\cite{BCC} introduced the idea of locally APN power functions as a generalization of the APN-ness property. A power function $F(X)$ over $\F_{2^n}$ is said to be locally-APN if $\max \{\delta_F (1, b) :b \in \F_{2^n} \setminus \F_2\} = 2.$ 

The boomerang attack on block ciphers is another important  cryptanalysis technique proposed by Wagner~\cite{DW}. It can be considered as an extension of the classical differential attack. At Eurocrypt 2018, Cid et al.~\cite{cid} introduced a systematic approach known as the Boomerang Connectivity Table
(BCT), to analyze the boomerang style attack. Boura and Canteaut~\cite{BoCa} further studied BCT and coined the term ``boomerang uniformity'', which is essentially the maximum value of nontrivial entries of the BCT, to quantify the resistance of a function against the boomerang attack. Boukerrou et al.~\cite{Bouk} pointed out the need for the counterpart of the BCT by extending the idea to Feistel ciphers. They introduced the Feistel Boomerang Connectivity Table (FBCT) as an extension of BCT for Feistel ciphers, where the S-boxes may not be permutations.  It is worth emphasizing that Feistel-based ciphers, such as 3-DES and CLEFIA~\cite{SSA}, hold just as much importance in the block cipher designs as SPN ciphers.

The authors in~\cite{Bouk} investigated the properties of the FBCT for two classes of functions, namely, APN functions and functions based on the inverse mapping over $\F_{2^n}$. Also, all the nontrivial coefficients of the FBCT are 0 for APN functions over $\F_{2^n}$ (via ~\cite{CV95}) and are 0 and 4 for the inverse function over $\F_{2^n}$, where $n$ is even. In fact, the coefficients of the FBCT are related to the second-order zero differential spectra of the functions. Another important property of the FBCT (sure, it can be derived from~\cite{CV95}) is that $F$ is an APN function over $\F_{2^n}$ if and only if the FBCT of $F$ is $0$ for $a,b \in \F_{2^n}$ with $ab(a+b) \neq 0$. Li et al.~\cite{LYT} further studied the second-order zero differential spectra of the inverse function and some APN functions in odd characteristic. The authors of~\cite{LYT} also show that these functions have low second-order zero differential uniformity. It is worth mentioning that Mesnager et al. gave an excellent survey on generalized differential and boomerang uniformities in \cite{MMM}. Moreover, Eddahmani and Mesnager \cite{EM} determined the FBCT entries of the inverse, the Gold, and the Bracken-Leander functions over finite fields of even characteristic. Their work also indicates that functions with low differential uniformity such as Bracken-Leander function may not have low second-order zero differential uniformity or $F$-boomerang uniformity. This suggests that, though in some cases (for instance APN in even characteristic, and PN in odd characteristic, as commented above), the coefficients of the FBCT are linked to the DU, this is not the case in general, and therefore, the FBCT is a truly independent measure of the strength of a function, from a cryptographic perspective. That is, low differential uniformity may not be enough, rather one has to combine it with many other criteria, including the $F$-boomerang uniformity. This motivates us further to investigate the second-order zero differential uniformity of functions that have low differential uniformity. It might seem like an unconnected list of functions, but our goal is to go through some of the known classes of functions with low differential uniformity and distinguish them via the $F$-boomerang uniformity and the FBCT.

Recently, Garg et al.~\cite{GHRS} and Man et al.~\cite{MLXZ,MMN} studied the second-order zero differential uniformity of several classes of functions over finite fields with low differential uniformity. We further extend their work by investigating the second-order zero differential spectra of some more classes of functions with low differential uniformity. In addition, these functions have low second-order zero differential uniformity. The paper is organized as follows. In Section~\ref{S2}, we recall some definitions. A connection between the FBCT, sum-freedom, and vanishing flats is established in Section~\ref{S30}. The second-order zero differential spectrum of four power functions over finite fields of odd characteristic is considered in Section~\ref{S3}. Further, in Section~\ref{S4}, the second-order zero differential spectrum and the number of vanishing flats of some power functions over finite fields of even characteristic are investigated. Section~\ref{S5} deals with second-order zero differential spectrum of two non-monomial permutations with low differential uniformity over finite field of even characteristic. Finally, we conclude the paper in Section~\ref{S6}. The techniques involve dealing with some delicate equations and systems of equations over finite fields and that may be of independent interest.

\section{Preliminaries}\label{S2}

In this section, we recall some definitions.
\begin{defn}~\textup{\cite{Bouk,LYT}} 
For $F : \F_{q} \rightarrow \F_{q}$, positive integer $n$, $p$ an arbitrary prime, and $a, b \in \F_{q}$, the {\em second-order zero differential spectra} of $F$ with respect to $a,b$ is defined as
\begin{equation*}
\nabla_F(a, b) := \# \{X \in \F_{q}: F(X + a + b) -F (X + b) -F(X + a) + F (X) = 0\}.
\end{equation*}
\end{defn}
If $q = 2^n$, we call $\nabla_F = \text{max}\{\nabla_F(a,b):a\neq b, a, b \in \F_{2^n} \setminus \{0\} \}$ the second-order zero differential uniformity of $F$. If $q=p^n,\,p > 2$, we call $\nabla_F =\text{max}\{\nabla_F(a,b) : a, b \in \F_{p^n}\setminus \{0\} \}$ the second-order zero differential uniformity of $F$.

\begin{defn}{(Feistel Boomerang Connectivity Table)}~\textup{\cite{Bouk}}
 Let $F:\F_{2^n}\to \F_{2^n}$ and $a,b \in \F_{2^n}$. The {\em Feistel Boomerang Connectivity Table (FBCT)} of $F$ is given by a $2^n \times 2^n$ table, in which the entry for the $(a, b)$ position is given by
 \begin{equation*}
\FB_F(a, b) =\# \{X \in \F_{2^n}: F(X + a + b) +F (X + b) +F(X + a) + F (X) = 0\}.
\end{equation*}
\end{defn}

\begin{defn}{($F$-Boomerang Uniformity)}~\textup{\cite{Bouk, LYT}}
The $F$-Boomerang uniformity corresponds to the highest value in the FBCT without considering the first row, the first column and the diagonal, that is,
$$
\beta_F = \max_{a \neq 0, b \neq 0, a \neq b} \FB_F(a,b).
$$
\end{defn}
Note that, when $p=2$, then  $\nabla_F(a, b)=\FB_F(a, b)$, and thus the coefficients of the FBCT are exactly the second-order zero differential spectra of functions over $\F_{2^n}$, and that the $F$-Boomerang uniformity is in fact the second-order zero differential uniformity of $F$ in even characteristic.

\begin{defn}
The Kloosterman sum $K(1)$ in the field $\F_{2^n}$ is defined by:
$$K(1) = \sum_{X \in \F_{2^n}} (-1)^{\Tr(X^{-1}+X)}$$
with the convention that $(-1)^{\Tr(X^{-1})} = 1$ for $X = 0$. Carlitz~\cite{Carlitz} gives the following explicit expression of this sum, namely,
$$K(1) = 1 + \dfrac{(-1)^{n-1}}{2^{n-1}} \sum_{i=0}^{\lfloor{n/2}\rfloor} (-1)^i \binom{n}{2i}7^i.$$
\end{defn}

We now present some lemmas that will be needed in the subsequent sections.

\begin{lem}~\textup{\cite{Bouk}}
\label{L1}
 Let $n \geq 1$ be an even integer. Let $\omega$ and $\omega + 1$ be the solutions of the equation $X^2 + X + 1$ in $\F_{2^n}$. Let $F$ be the inverse function defined over $\F_{2^n}$ by $F(0) = 0$ and $F(X) = \frac{1}{X}$ for $X \neq 0$. The FBCT of F satisfies
 \begin{equation*}
  \nabla_F(a,b)=
\begin{cases}
0 &~\mbox{~if~} a\neq0,~b \neq 0, ~a\neq b~and~a\not \in\{b\omega,b(\omega+1)\}\\
4 &~\mbox{~if~}  a\neq0,~b \neq 0, ~a\neq b~and~a\in\{b\omega,b(\omega+1)\}\\
2^n &~\mbox{~if~} ab=0~or~a=b.
\end{cases}
\end{equation*}
Moreover, $X\in \{0,a,b,a+b\}$ are the four solutions in the above system.
\end{lem}

\begin{lem}~\textup{\cite{Bouk}}\label{L2}
 Let $n \geq 1$ be an odd integer. Let $F$ be the inverse function defined over $\F_{2^n}$ by $F(0) = 0$ and $F(X) = \frac{1}{X}$ for $X \neq 0$. The FBCT of F satisfies
 \begin{equation*}
  \nabla_F(a,b)=
\begin{cases}
0 &~\mbox{~if~} a\neq0,~b \neq 0, ~a\neq b\\
2^n &~\mbox{~if~} ab=0~or~a=b.
\end{cases}
\end{equation*}
\end{lem}
Recall that for $p$ an odd prime, $n$ a positive integer, and $q=p^n$,  we let $\eta$ be the 
quadratic character of $\mathbb F_{q}$   defined by $\eta (X)=1$, respectively, $-1$, if $X$ is a square, respectively, a non-square in $\F_q^*$ (we let $\eta(0)=0$). We use $\#A$ for the cardinality of a set~$A$.
As customary, we shall use the convention that for any nonzero $a \in \F_{2^n}, a^{-1} := \frac{1}{a}$ and $0^{-1} := 0$, when needed.

\section{A connection between vanishing flats, sum-freedom, an alternative  definition of APNness for $p$ odd, and the FBCT}\label{S30}

In \cite{LMPPRS}, the concept of {\em vanishing flats} was introduced in the following way.
For $n \ge 2$, we define the set of all $2$-dimensional flats in $\F_{2^n}$ as
$$
\cB_n=\{ \{X_1,X_2,X_3,X_4\} \mid \mbox{$X_1+X_2+X_3+X_4=0$ and $X_1,X_2,X_3,X_4 \in\F_{2^n}$ are distinct} \}.
$$
Conventionally, each subset in $\cB_n$ is called a \emph{block}. For a function $F: \F_{2^n} \rightarrow \F_{2^n}$, we define the set of {\em vanishing flats} with respect to $F$ as
$$
\VB_{n,F}=\{ \{X_1,X_2,X_3,X_4\} \in \cB_n \mid F(X_1)+F(X_2)+F(X_3)+F(X_4)=0 \}.
$$

Clearly, given $\{X_1,X_2,X_3,X_4\}\in\VB_{n,F}$, we can always write $X_2=X_1+a,\,X_3=X_1+b,\,X_4=X_1+X_2+X_3=X_1+a+b$ for some $a,b\in\F_{2^n}$, and $X_1$ satisfies then the  equation $F(X + a + b) +F (X + b) +F(X + a) + F (X) = 0$. However, $X_2,X_3,X_4$ are also solutions of the equation, while the sets in the vanishing flats are unordered. Obviously, $\FB_F(a,a)=2^n$, and $\FB_F(a,b)=\FB_F(b,a)=\FB_F(a+b,a)=\FB_F(a,a+b)=\FB_F(a+b,b)=\FB_F(b,a+b)$. Thus, the following result holds.
\begin{prop}
\label{prop:vanish}
Given a function $F: \F_{2^n} \rightarrow \F_{2^n}$, $n\geq2$, then
$$\sum_{a,b\in\F_{2^n}^*,a\neq b}\FB_F(a, b) =24\,\# \VB_{n,F}.$$
\end{prop}

Therefore, the knowledge of the vanishing flats implies the knowledge of the Feistel Boomerang Connectivity Table (FBCT)/second-order zero differential spectra, though the vanishing flats shows not only the number of solutions but the actual solutions of the equation. However, the knowledge of the {\em number} of vanishing flats does not imply the knowledge of the Feistel Boomerang Connectivity Table (FBCT)/second-order zero differential spectra.

Very recently, Carlet~\cite{Carlet} introduced a generalization of APNness called {\em sum-freedom}. Let $2 \leq k \leq n$ and $m$ be positive integers. An $(n, m)$-function $F$ is called {\em $k$th-order sum-free} if, for every $k$-dimensional affine subspace (i.e. $k$-flat) $A$ of $\F_{2^n}$, the sum $\sum_{X \in A} F(X)$ is nonzero. Thus, $F$ is second-order sum free if and only if $\VB_{n,F}=\emptyset$, or equivalently,  $F$ is second-order sum free if and only if $\FB_F(a,b)=0$ for $a,b \in \F_{2^n}$ with $ab(a+b) \neq 0$.

In \cite{KT}, an alternative definition of APNness was given:  A function $F:\F_{p^n}\rightarrow\F_{p^n}$  is a {\em generalized almost perfect nonlinear (GAPN)}
function if the equation $$\tilde{D}_aF(X)=\sum_{i\in\F_p}F(X+ai)=b$$ has at most $p$ solutions $X$ in $\F_{p^n}$ if $a\in\F_{p^n}^*$ and $b\in\F_{p^n}$.

For $p=2$, this is clearly the notion of APNness, while for $p$ odd the derivative differs from the usual derivative for $p$ odd, which is given by ${D}_aF(X)=F(X)-F(X+a)$. It would be interesting to study how the functions studied in this paper for $p$ odd (and in general) behave under this generalized derivative. As far as the authors can see, there is no clear connection for $p$ odd between neither the differential uniformity nor the Feistel Boomerang Connectivity table.

\section{The second-order zero differential spectrum for power functions over finite fields of odd characteristic}\label{S3}
Let $p$ be now odd. Then, a function $F$ is PN  if and only if, for any $X_1,a\in\F_{p^n}$ ($a\neq0$), $$F(X_1+a)-F(Y+a)-F(X_1)+F(Y)=0$$ has exactly one solution, namely $Y=X_1$.

Similarly, $F$ is APN if and only if, for any $X_1,a\in\F_{p^n}$ ($a\neq0$), $$F(X_1+a)-F(Y+a)-F(X_1)+F(Y)=0$$ has at most  two distinct solutions $Y=X_1, X_2$, and there is at least one $X_1$ such that the equation has exactly two solutions $X_1\neq X_2$.

We see here that the study of the Feistel Boomerang Connectivity Table (FBCT)/second-order zero differential spectra is closely linked to how far is a function from the PN or the APN property, as is the case with the APN-ness in even characteristic.

Table~\ref{Table1} gives some of the known functions with low second-order zero differential uniformity over finite fields of odd characteristic, in addition to the general results on the PN/APN functions mentioned in~\cite{Bouk}.

\begin{table}[hbt]
\caption{Second-order differential uniformity for functions over finite fields of odd characteristic}
\label{Table1}
\begin{center}
\begin{tabular}{|c|c|c|c|c|c|} 
 \hline
 $p$ & $d$ & condition & $\Delta_F$ & $\nabla_F$ & Ref\\
 \hline
 $p >3$ & $3$ & any & $2$ & $1$ &   \textup{\cite[Theorem 3.1]{LYT}}\\ 
 \hline
 $p =3$ & $3^n-3$ & $n>1$ is odd& $2$ & $2$ &  \textup{\cite[Theorem 3.2]{LYT}}\\ 
 \hline
 $p >2$ & $p^n-2$ & $p^n \equiv 2\pmod 3$ & $2$& $1$ &  \textup{\cite[Theorem 3.3]{LYT}}\\ 
 \hline
 $p >3$ & $p^m+2$ & $n=2m$, $p^m \equiv 1\pmod 3$ & $2$ & $1$ &  \textup{\cite[Theorem 3.4]{LYT}}\\ 
 \hline
 $p=3$ & $3^n-2$ & any & $3$ & $3$ & \textup{\cite[Theorem 3.1]{LYT}}\\
 \hline
 $p$ & $p^n-2$ & $p^n \equiv 1\pmod 3$ & $3$ & $3$ & \textup{\cite[Theorem 3.1]{LYT}}\\
 \hline
  $p>3$ & $4$ & $n>1$ & $3$ & $2$ & This paper\\
 \hline
 $p$ & $\frac{2p^n-1}{3}$ & $p^n \equiv 2\pmod 3$  & $2$ & $1$ & This paper \\
 \hline
 $p>3$ & $\frac{p^k+1}{2}$ & $\gcd(2n,k)=1$ & $ \leq \gcd(\frac{p^k-1}{2},p^{2n}-1)$ & $\frac{p-3}{2}$& This paper \\
 \hline
 $p=3$ & $\frac{p^n-1}{2} + 2$ & $n$ is odd & $4$ & $3$ & This paper\\
 \hline
\end{tabular}
\end{center}
\end{table}
In this section, we first deal with the computation of the second-order zero differential spectrum of the function $F(X) = X^d$, where $d = \frac{2p^n-1}{3}$ over $\F_{p^n}$, for  $p^n \equiv 2\pmod 3$. Helleseth et al.~\cite{Hel} showed that $F$ is an APN function over $\F_{p^n}$, for  $p^n \equiv 2\pmod 3$.

\begin{thm}\label{T1}
 Let $F(X) = X^d$ be a function of $\F_{p^n}$, where $d = \frac{2p^n-1}{3}$, $p^n \equiv 2\pmod 3$. Then for $a,b\in \F_{p^n}$,
\begin{equation*} \nabla_F(a,b)=
\begin{cases}
1 &\text{~if~} ab\neq0\\
p^n &\text{~if~} ab=0.
\end{cases}
\end{equation*}
Moreover, $F$ is second-order zero differential $1$-uniform.
\end{thm}
\begin{proof} 
For $a,b \in \F_{p^n}$, we consider the equation
\begin{equation}
F(X + a + b) -F (X + b) -F(X + a) + F (X) = 0.
\end{equation}
If $ab = 0$, then $\nabla_F(a, b) = p^n$.
If $ab \neq 0$, then we can write the above equation as 
\begin{equation}\label{eqe2}
 \begin{cases}
  Y-X-a=0 \\
  F(Y+b)-F(Y)-F(X+b)+F(X)=0,
 \end{cases}
\end{equation}
or equivalently, substituting $Y'=\frac{Y}{b}$ and $X'=\frac{X}{b}$, we have
\begin{equation}\label {eqe11}
 \begin{cases}
  Y'-X'-ab^{-1}=0 \\
  F(Y'+1)-F(Y')-F(X'+1)+F(X')=0.
 \end{cases}
\end{equation}
Let $c \in \F_{p^n}^{*}$ and $F(Y'+1)-F(Y')=c$. Let $u_{Y'}=(Y')^d$ and $u_{Y'+1}=(Y'+1)^d$. Then $u_{Y'+1}=c+u_{Y'}$, and since $3d \equiv 1 $ mod $(p^n-1)$, we have $(u_{Y'+1})^3=(c+u_{Y'})^3$, or equivalently,
$$
3c(u_{Y'})^2+3c^2u_{Y'}+c^3-1=0.
$$
We get the quadratic equation in terms of $u_{Y'}$ and since $u_{Y'}=(Y')^d$, each of the two solutions gives a maximum of $\gcd(d,p^n-1)=1$ solutions in $Y'$. Solving the above quadratic equation, we have 
$$
u_{Y'}=\frac{-3c^2 \pm \sqrt{9c^4-4(3c)(c^3-1)}}{6c} = \frac{-3c^2 \pm \sqrt{-3c^4+12c}}{6c},
$$
and hence for $\Delta = \sqrt{-3c^4+12c}$, we have
$(u_{Y'})^3= Y' =  \left(\dfrac{-3c^2 \pm \Delta}{6c}\right)^3.$
This will give us $Y'_1 = \dfrac{\Delta^3-27c^6-9c^2\Delta^2+27c^4\Delta}{(216)c^3}$ and, $Y'_2 = \dfrac{-(\Delta^3+27c^6+9c^2\Delta^2+27c^4\Delta)}{(216)c^3}.$ Then $F(X'+1)-F(X')=c$ will also give us the same solutions 
 $X'_1 = \dfrac{\Delta^3-27c^6-9c^2\Delta^2+27c^4\Delta}{(216)c^3}$ and, $X'_2 = \dfrac{-(\Delta^3+27c^6+9c^2\Delta^2+27c^4\Delta)}{(216)c^3}.$ Now, from the first equation of System~\eqref{eqe11}, we have $Y_1'=X_2'+ab^{-1}$ or $Y_2'=X_1'+ab^{-1}$ as $ab^{-1} \neq 0$. We need to show that for every pair $(a,b) \in \F_{p^n}$, with $ab \neq 0$, there exists a unique $c$ for which either $Y_1'=X_2'+ab^{-1}$ or $Y_2'=X_1'+ab^{-1}$ holds. Now,
 $Y_1'-X_2'=ab^{-1}$ is the same as $2\Delta^3+54c^4\Delta=(216)c^3ab^{-1}$. Next, squaring the last expression we get, 
 \begin{equation}\label{eqq}
  4\Delta^6+(54)^2c^8\Delta^2+4(54)c^4\Delta^4=(216)^2c^6(ab^{-1})^2.
 \end{equation}
 We have, $\Delta^6 = (-3c^4+12c)^3=-(27)c^{12}+(12)^3 c^3-9(12)^2c^6+(27)(12)c^9$, and $\Delta^4 = (-3c^4+12c)^2 = 9c^8+ 144c^2-72c^5$. Substituting all the values, we have Equation~\eqref{eqq} reduced to 
 $$
 4c^{12}-12c^9+((27)(ab^{-1})^2-15)c^6-4c^3=0,
 $$
 or equivalently, 
 $$4k^{3}-12k^2+((27)(ab^{-1})^2-15)k-4=0,$$ 
 where $k=c^3$. We can write the above equation as below:
 \begin{equation}\label{eqr}
  (k-1)^3+\frac{(27)((ab^{-1})^2-1)k}{4}=0.
 \end{equation}

 When $(ab^{-1})^2=1$, that is $a = \pm b$, then $k=1$ and hence $c=1$, or we can say $\Delta=3$.  Thus we have $Y_1'=X_1'=0$ and  $Y_2'=X_2'=-1$. When $a=b$, $ab^{-1}=1$ and hence out of the pairs $(X_1',Y_2')$ and $(X_2',Y_1')$, we have $Y_1'-X_2'=1$ giving us $X=-b$ as a solution of Equation~\eqref{eqe2} in that case. And similarly for $a=-b$, we have $X=0$ as a solution of Equation~\eqref{eqe2}.

Let us now have $(ab^{-1})^2\neq1$. Substituting $k-1=Z$ and $A= \dfrac{(27)((ab^{-1})^2-1)}{4}$, we can write the above Equation~\eqref{eqr} as $Z^3+AZ+A=0$. The discriminant $D$ for the cubic equation is $D=-4A^3-(27)A^2=-A^2(4A+27)=-(27)(ab^{-1})^2 A^2$. Let us now assume that $Z^3+AZ+A=0$ has one root $Z_1 \in \F_{p^n}$. Then the cubic equation will have other two roots in $\F_{p^n}$ only if $D$ is a square in $\F_{p^n}$. Clearly $D$ is a square in $\F_{p^n}$ only if $\eta(-27)=\eta(-3)=1$, where $\eta$ is the quadratic character of $\F_{p^n}$. As $p^n \equiv 2 \pmod 3$, one can see that $\eta(-3)=-1$ (from \cite[Exercise 5.22]{LN}). Hence $D$ is not a square in $\F_{p^n}$. Thus, if $Z^3+AZ+A=0$ has a root in $\F_{p^n}$, then the other two roots are not in $\F_{p^n}$.

Thus we are left to show that  $Z^3+AZ+A=0$ is reducible over $\F_{p^n}$. However, that is immediate via a result of Dickson~\cite[Lemma 2]{Dickson06}, who showed that a necessary and sufficient condition for a cubic to have one and only one root in $\F_{p^n}$, $p>2$, is that the discriminant is a non-square in $\F_{p^n}$, which is what we showed above. Hence, for every pair $(a,b) \in \F_{p^n} \times \F_{p^n}$, the equation  $Z^3+AZ+A=0$ has exactly one solution $Z \in \F_{p^n}$ or equivalently a unique $k$, i.e. $k=Z+1$. Now, $\gcd(3,p^n-1)=1$ and hence there exists exactly one $c$ corresponding to each pair of $(a,b) \in \F_{p^n} \times \F_{p^n}$, for which $\Delta$ exists and $2\Delta^3+54c^4\Delta=(216)c^3ab^{-1}$, i.e. $Y_1'-X_2'=ab^{-1}$. Notice that $Y_1'-X_2'=ab^{-1}$ and $Y_2'-X_1'=ab^{-1}$ cannot hold simultaneously. Hence, the proof holds.
\end{proof}
Next, we consider the power function $F(X)= X^d$, where $d=\frac{p^k+1}{2}$, and compute its second-order zero differential spectrum over $\F_{p^n}$.
\begin{thm}\label{T2}
Let $F(X) = X^d$ be a power function of $\F_{p^n}$, where $d=\frac{p^k+1}{2}$, and $\gcd(k,2n)=1$. {Let $p>3$. }
Then for $a,b\in \F_{p^n}$,
\begin{equation*} 
\nabla_F(a,b)=
\begin{cases}
0 &~\text{~if~} ab \neq 0 \mbox{ and}~\eta(D)=-1 \\ 
1 &~\mbox{~if~} ab \neq 0 \mbox{ and}~  \eta(D)=0 \\
\frac{p-3}{2} &~\mbox{~if~} ab \neq 0 \mbox{ and}~\eta(D)=1\\
p^n &~\mbox{~if~} ab=0,
\end{cases}
\end{equation*}
where $D=\frac{4a^2}{(1-u^{2i})^2}+\frac{b^2}{u^{2i}}$, $u$ is a primitive $(p-1)$-th root of unity in $\F_{p^{2n}}^{*}$.  Moreover, $F$ is second-order zero differential $\frac{p-3}{2}$-uniform. 
\end{thm}
\begin{proof}
For $a,b \in \F_{p^n}$, we consider the equation
\begin{equation}\label{e2}
F(X + a + b) -F (X + b) -F(X + a) + F (X) = 0.
\end{equation}
If $ab = 0$, then $\nabla_F(a, b) = p^n$.
If $ab \neq 0$, then we split  Equation~\eqref{e2} in the following system
\begin{equation}\label{e3}
 \begin{cases}
  Y-X-a=0 \\
  F(Y+b)-F(Y)-F(X+b)+F(X)=0.
 \end{cases}
\end{equation}
 For $Y \in \F_{p^n}^{*}$, there exist $\alpha, \beta \in \F_{p^{2n}}^{*}$ such that $Y+ b = \alpha^2$ and $Y=\beta^2$. Let $\alpha-\beta= \theta \in \F_{p^{2n}}^{*}$. Then we have $\alpha+\beta= b \theta^{-1}$, giving us $\alpha= \dfrac{\theta + b \theta^{-1}}{2}$ and $\beta= \dfrac{b\theta^{-1}-\theta}{2}$. Thus,
 \allowdisplaybreaks
\begin{align*}
  (Y+b)^d-Y^d & =  \left(\dfrac{\theta + b \theta^{-1}}{2}\right)^{p^k+1} - \left(\dfrac{b\theta^{-1}-\theta}{2}\right)^{p^k+1}\\
  & = \dfrac{\theta^{p^k} (b \theta^{-1}) + (b \theta^{-1})^{p^k} \theta}{2}.
 \end{align*}

Using a similar argument as above we can write 
$$
(X+b)^d-X^d =  \dfrac{\lambda^{p^k} (b \lambda^{-1}) + (b \lambda^{-1})^{p^k} \lambda}{2},
$$
where $\lambda = \delta - \gamma \in \F_{p^{2n}}^{*}$, for $X+b = \delta^{2}$ and $X = \gamma^{2}$ such that $\delta, \gamma \in \F_{p^{2n}}^{*}$.
Now, using the second equation of  System~\eqref{e3}, we have 
$$
(Y+b)^d-Y^d = (X+b)^d-X^d,
$$
or equivalently, 
$$
\dfrac{\theta^{p^k} (b \theta^{-1}) + (b \theta^{-1})^{p^k} \theta}{2}  = \dfrac{\lambda^{p^k} (b \lambda^{-1}) + (b \lambda^{-1})^{p^k} \lambda}{2}.
$$
The above equality will give us either $(\lambda \theta^{-1})^{p^k-1}=1$ or $(\lambda \theta)^{p^k-1}=b^{p^k-1}$. Since $\gcd(p^{k}-1,p^{2n}-1)=p^{\gcd(k,2n)}-1=p-1$, we have either $\lambda = u^i \theta$ or $\lambda = b u^i \theta^{-1}$, where $u$ is a primitive $(p-1)$-th root of unity in $\F_{p^{2n}}^{*}$, $1 \leq i \leq p-1$. Now, we discuss all the possible values of $\lambda$. 

Note that, if $j=\frac{p-1}{2}+i$, then $u^j=-u^i$. We have then $\lambda =\pm u^i \theta$ or $\lambda =\pm b u^i \theta^{-1}$, where $u$ is a primitive $(p-1)$-th root of unity in $\F_{p^{2n}}^{*}$, $1 \leq i \leq \frac{p-1}{2}$.

If $\lambda = \pm \theta, \pm b \theta^{-1}$, then $X=Y$, but from the first equation of System~\eqref{e3}, we know $X \neq X+a$. Hence $\lambda \not \in  \{\pm \theta, \pm b \theta^{-1}\}$. 
 
In general, if $\lambda = \pm u^i\theta,\,\pm bu^i\theta^{-1}$, $u^i\neq1$, then 
\[
X=\left(\frac{b\lambda^{-1}-\lambda}{2}\right)^2=\frac{1}{4\lambda^2}(\lambda^4-2b\lambda^2+b^2)=\frac{1}{4u^{2i}\theta^2}(u^{4i}\theta^4-2bu^{2i}\theta^2+b^2).
\]
Thus,
 \allowdisplaybreaks
 \begin{align*}
 a&=Y-X=\frac{1}{4\theta^2}\left(\theta^4-2b\theta^2+b^2-\frac{u^{4i}\theta^4-2bu^{2i}\theta^2+b^2}{u^{2i}}\right),\\
4a\theta^2&=\theta^4(1-u^{2i})+b^2(1-u^{-2i}),\\
0&=\theta^4-\frac{4a}{1-u^{2i}}\theta^2+b^2\frac{1-u^{-2i}}{1-u^{2i}},\\
 0&=\theta^4-\frac{4a}{1-u^{2i}}\theta^2-b^2\frac{1}{u^{2i}},\\
\theta^2&=\frac{\frac{4a}{1-u^{2i}}\pm\sqrt{\frac{16a^2}{(1-u^{2i})^2}+4\frac{b^2}{u^{2i}}}}{2}=\frac{2a}{1-u^{2i}}\pm\sqrt{\frac{4a^2}{(1-u^{2i})^2}+\frac{b^2}{u^{2i}}}.
\end{align*}
When $D=\frac{4a^2}{(1-u^{2i})^2}+\frac{b^2}{u^{2i}}=0$, $\theta^2=\frac{2a}{1-u^{2i}}$, implying that $$X=\frac{1}{4}\left(u^{2i}\frac{2a}{1-u^{2i}}-2b+b^2u^{-2i}\frac{1-u^{2i}}{2a}\right).$$

When $\eta(D)=-1$, there is no solution. 
When $\eta(D)=1$, there are  two solutions for $X$,
$$X=\frac{1}{4}\left(u^{2i}\left(\frac{2a}{1-u^{2i}}\pm\sqrt{\frac{4a^2}{(1-u^{2i})^2}+\frac{b^2}{u^{2i}}}\right)-2b+b^2u^{-2i}\frac{1}{\frac{2a}{1-u^{2i}}\pm\sqrt{\frac{4a^2}{(1-u^{2i})^2}+\frac{b^2}{u^{2i}}}}\right).$$
Now, observe that 
\begin{align*}
\frac{1}{\frac{2a}{1-u^{2i}}+\sqrt{\frac{4a^2}{(1-u^{2i})^2}+\frac{b^2}{u^{2i}}}} &= \frac{-u^{2i}\left(\frac{2a}{1-u^{2i}}-\sqrt{\frac{4a^2}{(1-u^{2i})^2}+\frac{b^2}{u^{2i}}}\right)}{b^2},\text{ and}\\
\frac{1}{\frac{2a}{1-u^{2i}}-\sqrt{\frac{4a^2}{(1-u^{2i})^2}+\frac{b^2}{u^{2i}}}} &= \frac{-u^{2i}\left(\frac{2a}{1-u^{2i}}+\sqrt{\frac{4a^2}{(1-u^{2i})^2}+\frac{b^2}{u^{2i}}}\right)}{b^2}.
\end{align*}
Let us call 
$S_i = \{u^i \theta, u^j \theta, b u^i \theta^{-1}, b u^j \theta^{-1} \}$, where $j= \frac{p-1}{2}+i$ and $1\leq i \leq \frac{p-3}{2}$. Thus for $\lambda \in S_{i}$, we have the following two solutions,
 \allowdisplaybreaks
\begin{align*}
X_{i,1} & = \frac{1}{4}\left(u^{2i}\left(\frac{2a}{1-u^{2i}}+\sqrt{\frac{4a^2}{(1-u^{2i})^2}+\frac{b^2}{u^{2i}}}\right)-2b+b^2u^{-2i}\frac{1}{\frac{2a}{1-u^{2i}}+\sqrt{\frac{4a^2}{(1-u^{2i})^2}+\frac{b^2}{u^{2i}}}}\right)\\
& = \frac{1}{4}\left(u^{2i}\left(\frac{2a}{1-u^{2i}}+\sqrt{\frac{4a^2}{(1-u^{2i})^2}+\frac{b^2}{u^{2i}}}\right)-2b-\frac{2a}{1-u^{2i}}+\sqrt{\frac{4a^2}{(1-u^{2i})^2}+\frac{b^2}{u^{2i}}}\right)\\
& = \frac{1}{4}\left(-2a-2b+(1+u^{2i})\sqrt{\frac{4a^2}{(1-u^{2i})^2}+\frac{b^2}{u^{2i}}}\right)\\
& = \frac{1}{4}\left(-2a-2b+\frac{(1+u^{2i})}{(1-u^{2i})}\sqrt{{4a^2}-b^2 (1-u^{2i})(1-u^{-2i}))}\right)\\
& = \frac{1}{4}\left(-2a-2b+\frac{(1+u^{2i})}{(1-u^{2i})}\sqrt{{4a^2}+b^2 \frac{(u^{2i}-1)^2}{u^{2i}}}\right)
\end{align*}
and
 \allowdisplaybreaks
\begin{align*}
X_{i,2} & = \frac{1}{4}\left(u^{2i}\left(\frac{2a}{1-u^{2i}}-\sqrt{\frac{4a^2}{(1-u^{2i})^2}+\frac{b^2}{u^{2i}}}\right)-2b+b^2u^{-2i}\frac{1}{\frac{2a}{1-u^{2i}}-\sqrt{\frac{4a^2}{(1-u^{2i})^2}+\frac{b^2}{u^{2i}}}}\right)\\
& = \frac{1}{4}\left(u^{2i}\left(\frac{2a}{1-u^{2i}}-\sqrt{\frac{4a^2}{(1-u^{2i})^2}+\frac{b^2}{u^{2i}}}\right)-2b-\frac{2a}{1-u^{2i}}-\sqrt{\frac{4a^2}{(1-u^{2i})^2}+\frac{b^2}{u^{2i}}}\right)\\
& = \frac{1}{4}\left(-2a-2b-(1+u^{2i})\sqrt{\frac{4a^2}{(1-u^{2i})^2}+\frac{b^2}{u^{2i}}}\right)\\
& = \frac{1}{4}\left(-2a-2b-\frac{(1+u^{2i})}{(1-u^{2i})}\sqrt{{4a^2}-b^2 (1-u^{2i})(1-u^{-2i})}\right)\\
& = \frac{1}{4}\left(-2a-2b-\frac{(1+u^{2i})}{(1-u^{2i})}\sqrt{{4a^2}+b^2 \frac{(u^{2i}-1)^2}{u^{2i}}}\right).
\end{align*}
Now, one can see that $\frac{(u^{2k}-1)^2}{u^{2k}}=\frac{(u^{2l}-1)^2}{u^{2l}}$ for $1 \leq k \neq l \leq \frac{p-3}{2}$ implies that $u^{2k-2l}(u^{2k+2l}-1)=0$, or equivalently $u^{2k+2l}=1$. This will give us $u^k=-(u^l)^{-1}$, because being a primitive $(p-1)$-th root of unity, $u^{k+l} \neq 1$. Also, observe that for $u^k=-(u^l)^{-1}$, $\frac{(1+u^{2^k})}{(1-u^{2^k})}=-\frac{(1+u^{2^l})}{(1-u^{2^l})}$. The above discussion will give us $X_{k,1}=X_{l,2}$, when $u^k=-(u^l)^{-1}$.
Next, we note that 
$$
(u^{\frac{p-3}{2}-(t-1)})^{-1} = (u^{\frac{p-3}{2}})^{-1} u^{t-1} = u^{\frac{p-1}{2}} u^{t}= -u^{t},
$$
where $1 \leq t \leq \frac{p-3}{4}$ when $p \equiv 3 \pmod 4$
and $1 \leq t \leq \frac{p-5}{4}$ when $p \equiv 1 \pmod 4$ (notice that $(u^{\frac{p-1}{4}})^{-1}=-u^{\frac{p-1}{4}}$ when $p \equiv 1 \pmod 4$). Hence, for all the possible values of $i$ in our case, that is $1 \leq i \leq \frac{p-3}{2}$, we have $X_{t,1}=X_{\frac{p-3}{2}-(t-1),2}$ for $1 \leq t \leq i$. Thus for $(a,b) \in \F_{p^n} \times \F_{p^n}$ satisfying $\eta(D)=1$, Equation~\eqref{e2} has $\frac{p-3}{2}$ solutions in $\F_{p^n}$.
\end{proof}

\begin{rmk}
Helleseth et al. in~\textup{\cite{Hel}} showed that if $d=\frac{5^k+1}{2}$,  $\gcd(k,2n)=1$, then $F(X) = X^d$  is an APN power function over $\F_{5^n}$. Hence, from the above Theorem~\textup{\ref{T2}}, we get that $F$ is second-order zero differential $1$-uniform over $\F_{5^n}$. 
 Note also that, if $p=3$, then $F$ is PN, and therefore, by \cite{LYT}, it is second-order zero differential $0$-uniform over~$\F_{3^n}$. 
\end{rmk}

Now, we consider more functions with low (that is, $3,4$)  differential uniformity. In the following result, we show that $F(X)=X^4$ is second-order zero differential $2$-uniform for all $p,n >1$.
\begin{prop}\label{T3}
 Let $F(X) = X^4$ be a power function of $\F_{p^n}$, where $p$ is an odd prime and~$n>1$.
Then for $a,b\in \F_{p^n}$,
\begin{equation*} 
\nabla_F(a,b)=
\begin{cases}
0 &~\mbox{~if~} \eta\left(\dfrac{-a^2-b^2}{3}\right)=-1 \\
1 &~\mbox{~if~} a^2+b^2=0\\
2 &~\mbox{~if~} \eta\left(\dfrac{-a^2-b^2}{3}\right)=1 \\
p^n &~\mbox{~if}~ab=0.
\end{cases}
\end{equation*}
In particular, $F$ is second-order zero differential $2$-uniform.
\end{prop}
\begin{proof}
For $a,b \in \F_{p^n}$, we consider the equation: 
$$
F(X + a + b) -F (X + b) -F(X + a) + F (X) = 0.
$$
If $ab = 0$, then $\nabla_F(a, b) = p^n$.
If $ab \neq 0$, then we have the simplified equation as
$$X^2+(a+b)X+\dfrac{a^2+b^2}{3}+\dfrac{ab}{2}= 0.$$
One can easily observe that the discriminant of the above quadratic equation is $\dfrac{-(a^2+b^2)}{3}$.  This completes the proof.
\end{proof}

Helleseth et al. in~\cite{HS} showed that $F(X)=X^d$, where $d = \frac{3^{n}-1}{2}+ 2$ is a differentially $4$-uniform function over $\F_{p^n}$, for odd $n$. We compute its second-order zero differential spectrum implying that this function is second-order zero differential 3-uniform.
\begin{thm}\label{T4}
 Let $F(X) = X^d$ be a function on $\F_{3^n}$, where $d = \frac{3^{n}-1}{2}+ 2$ and $n$ is odd. Then for $a,b\in \F_{3^n}$,
\begin{equation*} 
\nabla_F(a,b)=
\begin{cases}
1 &~\mbox{~if~} \eta(ab)=1=\eta(a^2+b^2)\mbox{~or~}\eta(ab)=-1 \mbox{~and~} \eta(a^2+b^2)=1\\
3 &~\mbox{~if~} \eta(ab)=-1=\eta(a^2+b^2)\mbox{~or~}\eta(ab)=1 \mbox{~and~} \eta(a^2+b^2)=-1\\
3^n &~\mbox{~if}~ab=0.
\end{cases}
\end{equation*}
In particular,  $F$ is second-order zero differential $3$-uniform.
\end{thm}

\begin{proof}
 Now for $a,b \in \F_{3^n}$, we consider the equation: 
\begin{equation}\label{e1}
F(X + a + b) -F (X + b) -F(X + a) + F (X) = 0.
\end{equation}
If $ab = 0$, then $\nabla_F(a, b) = 3^n$.
If $ab \neq 0$, then we can write the above equation as 
\begin{equation*}
 \eta(X + a + b) (X + a + b)^2 -\eta(X + b)(X+b)^2 -\eta(X + a)(X+a)^2 +\eta(X)X^2 = 0,
\end{equation*}
where $\eta$ is the quadratic character over $\F_{3^n}$. Also $\eta(-1)=-1$, because $n$ is odd. Moreover, one can easily see that if $X$ is a solution of Equation~\eqref{e1}, then $-(X+a+b)$ is also a solution of  Equation~\eqref{e1}.
Now, for $X \not \in \{0,-a,-b,-(a+b)\}$, we have sixteen possible cases listed in Table~\ref{Table2}.
\begin{table}
\caption{Various cases considered when $ab \neq 0$}
\label{Table2}
\begin{center}
\scalebox{0.8}{
\begin{tabular}{|c|c|c|c|c|c|c|c|c|} 
 \hline
 Case&$\eta(X+a+b)$&$\eta(X+a)$&$\eta(X+b)$& $\eta(X)$ & Equation & X \\
 \hline 
 $1$&$1$ & $1$ & $1$ & $1$ & $2ab=0$ & No solution  \\
 \hline 
 $2$&$-1$ & $1$ & $1$ & $1$ & $X^2+2(a+b)X+a^2+b^2+ab=0$ & $-(a+b) \pm \sqrt{ab}$  \\
 \hline 
 $3$&$1$ & $-1$ & $1$ & $1$ &$X^2+2aX+a^2+ab=0$ & $-a \pm \sqrt{-ab}$   \\
 \hline 
 $4$&$1$ & $1$ & $-1$ & $1$ &$X^2+2bX+b^2+ab=0$&$-b \pm \sqrt{-ab}$\\
 \hline 
 $5$&$1$ & $1$ & $1$ & $-1$ & $X^2-ab=0$ & $\pm \sqrt{ab}$\\
 \hline 
 $6$&$-1$ & $-1$ & $1$ & $1$ & $2X+a+b=0$ & $X=\frac{-(a+b)}{2}$\\
 \hline 
 $7$&$-1$ & $1$ & $-1$ & $1$ & $2X+a+b=0$ & $X=\frac{-(a+b)}{2}$\\
 \hline 
 $8$&$-1$ & $1$ & $1$ & $-1$ & $2 X^2+2(a+b)X+a^2+b^2+ab=0$ & $\frac{-(a+b) \pm \sqrt{-(a^2+b^2)}}{2}$\\
 \hline 
 $9$&$1$ & $-1$ & $-1$ & $1$ & $2X^2+2(a+b)X+a^2+b^2+ab=0$& $\frac{-(a+b) \pm \sqrt{-(a^2+b^2)}}{2}$\\
 \hline 
 $10$&$1$ & $-1$ & $1$ & $-1$ & $2X+a+b=0$ & $X=\frac{-(a+b)}{2}$\\
 \hline 
 $11$&$1$ & $1$ & $-1$ & $-1$ & $2X+a+b=0$ & $X=\frac{-(a+b)}{2}$\\
 \hline 
 $12$&$-1$ & $-1$ & $-1$ & $1$ & $X^2-ab=0$ & $\pm \sqrt{ab}$\\
 \hline 
 $13$&$-1$ & $-1$ & $1$ & $-1$ & $X^2+2bX+b^2+ab=0$ & $-b \pm \sqrt{-ab}$ \\
 \hline 
 $14$&$-1$ & $1$ & $-1$ & $-1$ & $X^2+2aX+a^2+ab=0$ & $-a \pm \sqrt{-ab}$ \\
 \hline 
 $15$&$1$ & $-1$ & $-1$ & $-1$ & $X^2+2(a+b)X+a^2+b^2+ab=0$ & $-(a+b) \pm \sqrt{ab}$   \\
 \hline 
 $16$&$-1$ & $-1$ & $-1$ & $-1$ & $2ab=0$ &  No solution  \\
 \hline 
\end{tabular}
}
\end{center}
\end{table}
One can discard Case 1 and Case 16, as $ab \neq 0 $. Notice that Case 5 is not possible. This is because if $X=\sqrt{ab}$ is a solution with respect to Case 5, then we have $\eta(\sqrt{ab})=-1$ and $\eta(b+\sqrt{ab}) = \eta(a+\sqrt{ab}) = \eta(a+b+\sqrt{ab})=1$.
Then we can write
\begin{equation*}
\begin{split}
  1 &= \eta(b+\sqrt{ab}) \eta(a+\sqrt{ab}) 
   = \eta(\sqrt{ab}(a+b+2\sqrt{ab}))\\
  & = -\eta(a+b+2\sqrt{ab}) 
   = -\eta (a+b-\sqrt{ab}).
  \end{split}
\end{equation*}
This will give us $\eta(a+b+\sqrt{ab}) \eta (a+b-\sqrt{ab}) = \eta((a-b)^2)=-1$, which is not true. Hence, $X=\sqrt{ab}$ cannot be a solution of Case 5. Similarly, one can show that $X=-\sqrt{ab}$ is not a solution of Case 5. By similar reasoning, the possibility for Case 12 can also be removed. Also, Case 3, Case 4, Case 13 and Case 14 cannot hold. If $X = -a + \sqrt{-ab}$ is a solution of Case 3, then $\eta(-a + \sqrt{-ab})=1=\eta(b-a + \sqrt{-ab})=\eta(b + \sqrt{-ab})$ and $\eta(\sqrt{-ab})=-1$. Thus, we have
\begin{equation*}
  1 = \eta(b+\sqrt{-ab}) \eta(-a+\sqrt{-ab}) 
   = \eta(\sqrt{-ab}(b-a+2\sqrt{ab}))
   = -\eta(b-a-\sqrt{ab}).
\end{equation*}
This renders $\eta(b-a+\sqrt{-ab})\eta(b-a-\sqrt{-ab})=\eta((a+b)^2)=-1$, which is not true.

Now we split our analysis in four cases depending on the values of $\eta(ab)$ and $\eta(a^2+b^2).$

\textbf{Case A.} Let $\eta(ab)=1$ and $\eta(a^2+b^2)=1$. We can have only one solution possible in this case. If Case 2 holds, then $X_1 = -(a+b) \pm \sqrt{ab}$ is a solution of Equation~\eqref{e1} for  Case A. But then $-(X_1+a+b) = \pm \sqrt{ab}$ is also a solution which is not possible. Similarly, Case 15 would also not contribute to Case A. Hence the only possible solution we have is either obtained from Case 10 or Case 11. 

\textbf{Case B.} Let $\eta(ab)=-1$ and $\eta(a^2+b^2)=1$. We have exactly one solution in $\F_{3^n}$ of Equation~\eqref{e1}, i.e. $X=\frac{-(a+b)}{2}$, which is obtained either from Case 10 or Case 11.

\textbf{Case C.} Let $\eta(ab)=1$ and $\eta(a^2+b^2)=-1$. This case will have three solutions. One solution would be obtained either from Case 10 or Case 11. Notice that Case 8 and Case 9 can have at most one solution individually. Also, they cannot occur simultaneously. Hence, if $X_1 = \frac{-(a+b) + \sqrt{-(a^2+b^2)}}{2}$, a solution obtained from Case 8 satisfies Equation~\eqref{e1}, then $Y_1=-(X_1+a+b) = \frac{-(a+b) - \sqrt{-(a^2+b^2)}}{2}$ is also a solution of Equation~\eqref{e1}. Observe that $Y_1$ can be obtained from Case 9. Similarly we have one more possible pair of solutions, i.e.   
$X_2 =  \frac{-(a+b) - \sqrt{-(a^2+b^2)}}{2}$ and $Y_2 =  \frac{-(a+b) + \sqrt{-(a^2+b^2)}}{2}$ satisfying Case 8 and Case 9 respectively. Summarizing all, we have exactly three solutions in Case C.

\textbf{Case D.} Let $\eta(ab)=-1$ and $\eta(a^2+b^2)=-1$. Using similar arguments used in Case C, one can show that there are three solutions in this case as well.

This completes the proof.
\end{proof}

\section{The second-order zero differential spectrum for power functions over finite fields of even characteristic}\label{S4}

In this section, we first compute the second-order zero differential spectrum of the family of functions $F(X)=X^{2^{t}-1}$ over $\F_{2^{n}}$, whose DDT entries have already been computed by Blondeau et al. in~\cite{BCC}. We then compute the number of vanishing flats. Finally, we consider a few concrete values of $t$ that give low differential uniformity.

Let $F(X)=X^d$ be a monomial. Recall that for any $(a',b') \in \F_{2^n}^* \times \F_{2^n}^*$, $\nabla_F(a',b')$ is given by the number of solutions $X \in \F_{2^n}$ of the following equation 
\begin{equation}\label{F2}
X^{d}+(X+a')^{d}+(X+b')^{d}+(X+a'+b')^{d}=0.
\end{equation}
Observe that for $a'=0,b'=0$ and $a'=b'$,  Equation~\eqref{F2} has $2^n$ solutions in $\F_{2^n}$. Let $a',b'\neq0$. We now rewrite Equation~\eqref{F2} as

\begin{equation*} 
x^{d}+(x+1)^{d}+(x+b)^{d}+(x+b+1)^{d}=0,
\end{equation*}
where $b=\frac{b'}{a'}$, $x=\frac{X'}{a'}$ and $a' \neq 0$. We can, therefore, assume without loss of generality that $a'=1$, since $\nabla_F(a',b')=\nabla_F(1,b'/a')$ if $a'\neq0$. Moreover, we know that for a monomial $F(X)=X^d$, the DDT entries of $F$ are determined by $\delta_F(1,b), b \in \F_{2^n}$. Throughout this paper, we shall denote $\delta_F(1,b)$ by $\delta_F(b)$, and the differential uniformity of $F$ by $\delta_F$.

\begin{thm} 
\label{thmt} 
Let $F$ be a function over $\F_{2^n}$ defined by $F(X)=X^{2^t-1}$. Then, for any $b \in \F_{2^n}$,  
$$\nabla_F(1,b)=\left\{\begin{array}{ll}
2^n &\text{if }b=0,1,\\
\delta_F(B) &\text{if }B=1,\\
\max(\delta_F(B)-2,0) &\text{if } B\neq1,\\
0 & \mbox{otherwise~},\\
\end{array}\right.=\left\{\begin{array}{ll}
2^n &\text{if }b = 0,1,\\
2^{\gcd(t-1,n)} &\text{if }B=1,\\
2^{\gcd(t,n)}-4 &\text{if }B=0,\\
\max(2^r-4,0) &\text{if }B\neq0,1, \\
0 & \mbox{otherwise},
\end{array}\right.$$
where $B= \dfrac{b^{2^t}+b}{b(b+1)}$ for $b\neq0,1$, $1< r\leq \min(t,n-t+1)$, 
and $r$ depends on $B$, precisely, $r=\dim(Ker(P_B))$, $P_B=X^{2^t}+BX^2+(B+1)$. 
\end{thm}

NB: Note that $\delta_F(1)=2^{\gcd(t-1,n)}$ and $\delta_F(0)=2^{\gcd(t,n)}-2$, by \cite{BCC}. However, given that the relation between $b$ and $B$ is in general not a permutation, $\delta_F(B)\leq \delta_F$, the differential uniformity, but $\max_{b}\delta_F(B)$ is not necessarily equal to $\delta_F$. Therefore, the $F$-Boomerang uniformity is bounded by the differential uniformity: $$
\beta_F = \max_{a \neq 0, b \neq 0, a \neq b} \FB_F(a,b)\leq   \delta_F,$$
but the equality may not happen.

\begin{proof}
Recall that for any $(a',b') \in \F_{2^n}^* \times \F_{2^n}^*$, $\nabla_F(a',b')$ is given by the number of solutions $X \in \F_{2^n}$ of the following equation 
\begin{equation}\label{Ft}
X^{2^t-1}+(X+a')^{2^t-1}+(X+b')^{2^t-1}+(X+a'+b')^{2^t-1}=0.
\end{equation}
Observe that for $a'=0,b'=0$ and $a'=b'$,  Equation~\eqref{Ft} has $2^n$ solutions in $\F_{2^n}$. We now rewrite the Equation~\eqref{Ft} as
\begin{equation}\label{Ft2}
x^{2^t-1}+(x+1)^{2^t-1}+(x+b)^{2^t-1}+(x+b+1)^{2^t-1}=0,
\end{equation}
where $x=\dfrac{X}{a'}$ and $b=\dfrac{b'}{a'}$. Clearly, for $b=0$ and $b=1$, Equation~\eqref{Ft2} has $2^n$ solutions in $\F_{2^n}$. We assume $b \in \F_{2^n}^*\setminus \{1\}$. Then we have the following two cases:

\textbf{Case 1.} Let $x \in \{ 0,1,b,b+1 \}$. Then  Equation~\eqref{Ft2} reduces to
\begin{equation*}
 1+b^{2^t-1}+(1+b)^{2^t-1}=0,\text{ that is, }  b^2+b^{2^t}=0.
\end{equation*}
This is equivalent to $b^{2^t-2}=1$.

\textbf{Case 2.} Let $x \not \in \{ 0,1,b,b+1 \}$. Then  Equation~\eqref{Ft} reduces to
\begin{equation*}
\frac{x^{2^t}}{x}+\frac{(x+1)^{2^t}}{x+1}+\frac{(x+b)^{2^t}}{x+b}+\frac{(x+1+b)^{2^t}}{x+1+b}=0,
\end{equation*}
 or equivalently,
 \begin{equation*}
b(b+1) x^{2^t} + (b^{2^t}+b)x^2+(b^{2^t}+b^2)x=0.
\end{equation*}
As $b \in \F_{2^n}^*\setminus \{1\}$, we can rewrite the above equation as
\begin{equation}\label{Ft3}
P_B(x)=x^{2^t} + B x^2+(B+1) x =0,
\end{equation}
where $B= \dfrac{b^{2^t}+b}{b(b+1)}$. Then, we need to compute the number of solutions to Equation~\eqref{Ft3}. By \cite[Theorem 1]{BCC},  
$\delta_F(0)=2^{\gcd(t,n)}-2,\delta_F(1)=2^{\gcd(t-1,n)},\delta_F(b)=2^r-2,b\in\F_{2^n}\setminus\F_2$, $1\leq r\leq \min(t,n-t+1)$. Observe that, when $r=1$, $\delta_F(b)=0$, and that $\delta_F(b)$ is the number of solutions other than 0,1, for $b\neq1$. The result then follows.

Note that, since $B=0$ if and only if $b^{2^t-1}=1$, and this is only possible for $b\neq1$ if $\gcd(2^t-1,2^n-1)=\gcd(t,n)>1$, then in this case $2^{\gcd(t,n)}-4\geq0$, giving that  $\max(2^{\gcd(t,n)}-4,0)=2^{\gcd(t,n)}-4$.
\end{proof}

From the above result and Proposition~\ref{prop:vanish}, we determine the number of vanishing flats of $F(X)=X^{2^t-1}$ over $\F_{2^{n}}$, as stated in the following corollary. We point to~\cite{LMPPRS}, where some of the vanishing flat counts are also found. A good summary can be found in \cite[Table III.1]{LMPPRS}, where we find the count for $t=n-1,n-2, \frac{n-1}{2},\frac{n+3}{2}$ ($n$ odd), $\frac{n}{2},\frac{n}{2}+1$ ($n$ even). Our next corollary considers arbitrary values of~$t$.

\begin{cor}\label{R1}
Let  $F(X)=X^{2^t-1}$ over $\F_{2^{n}}$, then the number of vanishing flats $\VB_{n,F}$ of $F$ is given by 
$$\frac{1}{24}\left [(2^n-1)(4^{\gcd(t-1,n)}+4^{\gcd(t,n)} -3\cdot 2^{\gcd(t,n)+1} -2^{\gcd(t-1,n)+1}+8)+(2^{r}-4)(\# \mathcal R_{2^{r}-4}) \right],
$$
where $\# \mathcal R_{2^{r}-4} = (2^n-1)\cdot \# \left\{b \in \F_{2^n}^{*}\setminus \{1\} \mid  b^{2^t-2} \neq 1, b^{2^t-1} \neq 1, \delta_F\left ( \frac{b^{2^t}+b}{b(b+1)}\right)=2^r-2 \right \}$, $1< r \leq \min(t,n-t+1)$, $r (=r(B))=\dim(Ker(P_B))$, $P_B=X^{2^t}+BX^2+(B+1)$.
\end{cor}
\begin{proof}
 From Theorem~\ref{thmt}, we know that for any $b \in \F_{2^n}^{*}$,
 $$\nabla_F(1,b)=\left\{\begin{array}{ll}
2^n&\text{if }b=0,1,\\
\delta_F(B)&\text{if }B=1,\\
\max(\delta_F(B)-2,0)&\text{if }B\neq1,\\
0 & otherwise,
\end{array}\right.=\left\{\begin{array}{ll}
2^n&\text{if }b=0,1,\\
2^{\gcd(t-1,n)}&\text{if }B=1,\\
2^{\gcd(t,n)}-4&\text{if }B=0,\\
\max(2^r-4,0)&\text{if }B\neq0,1,\\
0 & otherwise,
\end{array}\right.$$
where $B= \dfrac{b^{2^t}+b}{b(b+1)}$, and $1< r\leq \min(t,n-t+1)$.
We consider the set $$\mathcal S_{i} =\{b \in \F_{2^n}^{*} \setminus \{1\} \mid \FB_F(1,b)=i\}.$$
It is obvious that $B=1$ if and only if $b^{2^t-2}=1$ for any $b \in \F_{2^n}^{*}$. Therefore, we get $\# \mathcal S_{2^{\gcd(t-1,n)}} = \gcd(2^t-2,2^n-1)-1=2^{\gcd(t-1,n)}-2.$ Similarly, $B=0$ if and only if $b^{2^t-1}=1$ for any $b \in \F_{2^n}^{*}$. This will give us $\# \mathcal S_{2^{\gcd(t,n)}-4} = \gcd(2^t-1,2^n-1)-1=2^{\gcd(t,n)}-2.$

Next, we need to consider the case when $B \notin \{0,1\}$ and $\FB_F(1,b)=2^r-4$ where $1< r\leq \min(t,n-t+1)$. Here, we will have $\# \mathcal S_{2^{r}-4} = \#\{b \in \F_{2^n}^{*}\setminus \{1\} \mid b^{2^t-2} \neq 1, b^{2^t-1} \neq 1, \delta_F\left ( \dfrac{b^{2^t}+b}{b(b+1)}\right)=2^r-2 \}$, $1< r\leq \min(t,n-t+1)$. Moreover, if for each $a \in \F_{2^n}^{*}$,
$$\mathcal R_{i} =\{(a,b) \in \F_{2^n}^{*} \times \F_{2^n}^{*} \setminus \{a\} \mid \FB_F(a,b)=i\},$$ 
then
\begin{align*}
\# \mathcal R_{2^{\gcd(t-1,n)}} &= (2^{\gcd(t-1,n)}-2)(2^n-1),\\
\# \mathcal R_{2^{\gcd(t,n)}-4} &= (2^{\gcd(t,n)}-2)(2^n-1),
\end{align*}
and
$\# \mathcal R_{2^{r}-4} = (2^n-1)\left(\# \left\{b \in \F_{2^n}^{*}\setminus \{1\} \mid  b^{2^t-2} \neq 1, b^{2^t-1} \neq 1, \delta_F\left ( \frac{b^{2^t}+b}{b(b+1)}\right)=2^r-2 \right \}\right),$ where $1< r \leq \min(t,n-t+1).$

Thus, from Proposition~\ref{prop:vanish}, the number of vanishing flats $\VB_{n,F}$ of $F$ is given by
 \begin{align*}
& \frac{1}{24} \left [ 2^{\gcd(t-1,n)}(2^{\gcd(t-1,n)}-2)(2^n-1) \right.\\
 &\qquad\qquad\qquad \left. +(2^{\gcd(t,n)}-4)(2^{\gcd(t,n)}-2)(2^n-1)+(2^{r}-4)(\# \mathcal R_{2^{r}-4})\right],
 \end{align*}
 where $1< r\leq \min(t,n-t+1)$, which renders the claim.
\end{proof}

We will next consider some particular classes with low differential uniformity. Firstly, we compute the second-order zero differential spectrum of the locally APN function $F_1(X)=X^{2^{m}-1}$ over $\F_{2^{2m}}$, the DDT entries for which have already been computed by Blondeau et al. in~\cite{BCC}. We then consider a few more permutations with low differential uniformity.
\begin{cor}\label{T5}
Let  $F_1(X) = X^{2^m-1} \in \F_{2^n}[X]$, where $n=2m$. Then for any $b \in \F_{2^n}$, we have:
\begin{enumerate}
 \item[$(1)$] When $m$ is odd,
 \begin{equation*}\nabla_{F_1}(1,b)=
  \begin{cases}
   \ 2^n &\text{\rm if~}  b=0, \text{\rm ~or~} b=1, \\
   \ 4  &\text{\rm if~}  b\neq0,1 \text{\rm ~and~}  b^{2^{m}-2}=1, \\
   \ 2^m-4  &\text{\rm if~}  b\neq0,1 \text{\rm ~and~} b^{2^{m}-1}=1, \\
   \ 0  & \text{\rm otherwise}.
  \end{cases}
 \end{equation*}
\item[$(2)$] When $m$ is even,
\begin{equation*} \nabla_{F_1}(1,b)=
  \begin{cases}
   \ 2^n &\text{\rm if~}  b=0, \text{\rm ~or~} b=1, \\
   \ 2^m-4  &\text{\rm if~}  b\neq0,1 \text{\rm ~and~}b^{2^{m}-1}=1, \\
   \ 0  &\text{\rm otherwise.} 
  \end{cases}
 \end{equation*}
\end{enumerate}
Thus, the Feistel boomerang uniformity of $F_1$ is $\beta_{F_1} = 2^m-4$ for $m$ even or $m$ odd with $m>2$.
\end{cor}

\begin{proof}
Here, $t=m=\frac{n}{2}$. Then, $\gcd(t-1,n)=\gcd(m-1,2m)=1$, when $m$ is even and $\gcd(t-1,n)=\gcd(m-1,2m)=2$ when $m$ is odd, respectively. Moreover, $\gcd(t,n)=\gcd(m,2m)=m$ for every positive integer $m$. Let $B= \dfrac{b^{2^m}+b}{b(b+1)}$ for $b\neq0,1$. Furthermore, from Theorem 7 in ~\cite{BCC}, $\delta_{F_1}(0)=2^m-2$, $\delta_{F_1}(1)=4$ when $m$ is odd and $\delta_{F_1}(1)=2$ when $m $ is even.  Also, $\delta_{F_1}(B) \leq 2$ for $B \neq 0,1$. This implies that $\max(\delta_{F_1}(B)-2,0)=0$ for $B \neq 0,1$. Hence from Theorem~\ref{thmt}, we have the following:
\begin{enumerate}
 \item If $m$ is even,
 $$
 \nabla_{F_1}(1,b)=\left\{\begin{array}{ll}
2^n &\text{ if }b=0,1,\\
\delta_{F_1}(B) &\text{ if }B=1,\\
\max(\delta_{F_1}(B)-2,0) &\text{ if }B\neq1,\\
\ 0 & \text{ otherwise}.
\end{array}\right.=\left\{\begin{array}{ll}
2^n&\text{ if }b=0,1,\\
2&\text{ if }B=1,\\
2^m-4&\text{ if }B=0,\\ 
\ 0 & \text{ otherwise}.
\end{array}\right.
$$
\item If $m$ is odd,
$$
\nabla_{F_1}(1,b)=\left\{\begin{array}{ll}
2^n &\text{ if }b=0,1,\\
\delta_{F_1}(B) &\text{ if }B=1,\\
\max(\delta_{F_1}(B)-2,0)&\text{ if }B\neq1,\\
\ 0 & \text{ otherwise}.
\end{array}\right.=\left\{\begin{array}{ll}
2^n&\text{ if }b=0,1,\\
4&\text{ if }B=1,\\
2^m-4&\text{ if }B=0,\\
\ 0 & \text{ otherwise}.
\end{array}\right.
$$
\end{enumerate}
 Clearly, 
 $$
 B=1\Leftrightarrow B+1= \dfrac{b^{2^m}+b^2}{b(b+1)}=0\Leftrightarrow b^{2^m-2}=1,b\neq1.
 $$
Now, $\gcd(2^m-2,2^{2m}-1)=1$ if $m$ is even, and  $\gcd(2^m-2,2^n-1)=3$ if  $m$ is odd, implying that $B=1$ is not possible for $m$ even, while there are exactly two values of $b$ such that $B=1$ if $m$ is odd.

Thus the proof follows immediately from $\nabla_{F_1}(a,b)=\nabla_{F_1}(1,b/a)$.
\end{proof}

\begin{cor}
Let  $F_1(X)=X^{2^m-1}$ over $\F_{2^{2m}}$, then we have
\begin{equation*}\# \VB_{n,F_1}=
  \begin{cases}
   \  \dfrac{(2^{m-2}-1)(2^{m-1}-1)(2^n-1)}{3} &\text{\rm ~if~$m$~is~even,} \\
   \ \dfrac{((2^{m-2}-1)(2^{m-1}-1)+1)(2^n-1)}{3} &\text{\rm ~if~$m$~is~odd.} 
  \end{cases}
\end{equation*}
\end{cor}
\begin{proof}

The number of vanishing flats for $F_1(X)=X^{2^m-1}$ over $\F_{2^{2m}}$ follows from Corollary~\ref{R1} and Corollary~\ref{T5}. 
\end{proof}

Blondeau et al.~\cite{BCC} also studied the differential uniformity of $F_2(X) = X^{2^{\frac{n-1}{2}}-1}$  for $n$ odd, and showed that it is a permutation having differential uniformity at most 8. In the following corollary, we determine its FBCT spectrum and show that its $F$-boomerang uniformity is at most~8 for $m \not \equiv 1 \pmod 3$, and is exactly 4 for $m \equiv 1 \pmod 3$.
\begin{cor}\label{T51}
Let  $F_2(X) = X^{2^m-1} \in \F_{2^n}[X]$, where $n=2m+1$. Then for any $b \in \F_{2^n}$, we have:
\begin{enumerate}
 \item[$(1)$] When $m \equiv 1 \pmod 3$,
\begin{equation*}\nabla_{F_2}(1,b)=
\begin{cases}
   \ 2^n &\text{\rm ~if~}  b=0, \text{\rm ~or~} b=1, \\
   \ 8  &\text{\rm ~if~}  b\neq 0,1,  \text{\rm ~and~} b^{2^{m}-2}=1, \\
   \ 4  &\text{\rm ~if~}  b \in \mathcal S_6, \\
   \ 0&\text{\rm ~otherwise.~}
  \end{cases}
 \end{equation*}
\item[$(2)$] When $m \not \equiv 1 \pmod 3$,
\begin{equation*} \nabla_{F_2}(1,b)=
  \begin{cases}
   \ 2^n &\text{\rm ~if~}  b=0, \text{\rm ~or~} b=1, \\
   \ 4  &\text{\rm ~if~}  b \in \mathcal S_6,\\
  \ 0&\text{\rm ~otherwise,~}
  \end{cases}
 \end{equation*}
\end{enumerate}
where $S_6=\{b \in \F_{2^n}^{*}\setminus \{1\},  b^{2^m-2} \neq 1, b^{2^m-1} \neq 1$ \mbox{~and~} $\delta_{F_2}\left ( B\right)=6\}$, with $B=\dfrac{b^{2^m}+b}{b(b+1)}$. Thus, the Feistel boomerang uniformity of $F_2$ is $8$ for $m  \equiv 1 \pmod 3$, and it is $4$ for $m \not\equiv 1 \pmod 3$.
\end{cor}
\begin{proof}
Notice that $t=m=\frac{n-1}{2}$. Then, $\gcd(t-1,n)=\gcd(m-1,2m+1)=3$ if $m \equiv 1 \pmod 3$ and $1$ if $m \not \equiv 1 \pmod 3$, respectively. Additionally, $\gcd(t,n)=\gcd(m,2m+1)=1,$ for all positive integers $m$. For $B= \dfrac{b^{2^m}+b}{b(b+1)}$ where $b\neq0,1$, observe from~\cite[Theorem 9]{BCC} that $\delta_{F_2}(1)=8$ when $m \equiv 1 \pmod 3$ and otherwise, $\delta_{F_2}(1)=2$. Moreover, we have $\delta_{F_2}(0)=0$ and $\delta_{F_2}(B)=6$ for some $B \in \F_{2^n}^{*}\setminus \{1\}$. This further implies that $\max(\delta_{F_2}(B)-2,0)=4$ when $B \neq 0,1$. Therefore from Theorem~\ref{thmt}, we get:
\begin{enumerate}
 \item If $m \equiv 1 \pmod 3$,
 $$\nabla_{F_2}(1,b)=\left\{\begin{array}{ll}
2^n &\text{if }b=0,1,\\
\delta_{F_2}(B) &\text{if }B=1,\\
\max(\delta_{F_2}(B)-2,0) &\text{if }B\neq1,\\
\ 0&\text{\rm otherwise}
\end{array}\right.=\left\{\begin{array}{ll}
2^n&\text{if }b=0,1,\\
8 &\text{if }B=1,\\
0 &\text{if }B=0,\\
4 &\text{if }B\neq0,1,\\
0&\text{\rm otherwise.~}
\end{array}\right.$$
\item If $m \not \equiv 1 \pmod 3$,
$$\nabla_{F_2}(1,b)=\left\{\begin{array}{ll}
2^n&\text{if }b=0,1,\\
\delta_{F_2}(B) &\text{if }B=1,\\
\max(\delta_{F_2}(B)-2,0)&\text{if }B\neq1,\\
 0&\text{\rm otherwise.~}
\end{array}\right.=\left\{\begin{array}{ll}
2^n&\text{if }b=0,1,\\
2&\text{if }B=1,\\
0&\text{if }B=0,\\
4 &\text{if }B\neq0,1,\\
0&\text{\rm otherwise.~}
\end{array}\right.$$
\end{enumerate}
Since, $\gcd(2^m-1,2^n-1)=1$, $B=0$ is not possible. Moreover, $B=1$ if and only if $b^{2^m-2}=1, b\neq1.$ Clearly, $\gcd(2^m-2,2^n-1)=1$ when $m \not \equiv 1 \pmod 3$, so in that case $B=1$ is not possible. We then need to consider only the case  $m \equiv 1 \pmod 3$; in this case $\gcd(2^m-2,2^n-1)=7$, giving us exactly six values of $b \in \F_{2^n}^{*} \setminus \{1\}$ such that $B=1$.

Next, we have $\nabla_{F_2}(1,b)=4$, when $\delta_{F_2}(B)=6$ and $B\neq0,1$. More precisely, we can say that $\nabla_{F_2}(1,b)=4$ when $b \in \mathcal S_6 =\{b \in \F_{2^n}^{*}\setminus \{1\},  b^{2^t-2} \neq 1, b^{2^t-1} \neq 1$ \mbox{~and~} $\delta_{F_2}\left ( \dfrac{b^{2^t}+b}{b(b+1)}\right)=6\}$.
\end{proof}

In the next corollary, we give the number of vanishing flats of $F_2$.
\begin{cor}\label{R2}
Let  $F_2(X)=X^{2^m-1}$ over $\F_{2^{2m+1}}$, then we have
\begin{equation*}\# \VB_{n,F_2}=
  \begin{cases}
   \  \left(\dfrac{2^{n-2}+7}{6}-\dfrac{K(1)}{8}\right)(2^n-1) &\text{\rm ~if~} m \equiv 1\pmod 3\\
   
   \  \left(\dfrac{2^{n-2}+1}{6}-\dfrac{K(1)}{8}\right)(2^n-1) &\text{\rm ~if~} m \not \equiv 1\pmod 3,
  \end{cases}
\end{equation*}
where $K(1)$ is the Kloosterman sum.
\end{cor}
\begin{proof} 
From Corollary~\ref{T51}, we have $\nabla_{F_2}(1,b)=4$, when $b \in S_6$. The cardinality of the set $S_6$ is  
\begin{align*}
 \# S_6 & = \# \left\{b \in \F_{2^n}^{*}\setminus \{1\} \,|\, b^{2^t-2} \neq 1, b^{2^t-1} \neq 1 \mbox{~and~}\delta_{F_2}(B)= \delta_{F_2}\left ( \dfrac{b^{2^t}+b}{b(b+1)}\right)=6\right\}\\
 & = \# \{b \in \F_{2^n}^{*}\setminus \{1\}  \,|\,  b^{2^t-2} \neq 1, b^{2^t-1} \neq 1 \\ 
 &\qquad\qquad\qquad\qquad \mbox{~and~} x^{2^m}+Bx^2+(B+1)x \mbox{~has~six~solutions~}\in \F_{2^n}\}.
\end{align*}
It is already known from  \cite[Theorem 5]{BP} that for $m \equiv 1 \pmod 3$, 
$$\# \{B \in \F_{2^n}^{*}\setminus \{1\} \mbox{~and~}\delta_{F_2}(B)=6 \}= \dfrac{2^{n-2}-5}{6}-\dfrac{K(1)}{8}.$$
We thereby only need to count the number of $b \in \F_{2^n}^{*}\setminus \{1\}$ such that $\delta_{F_2}(B)=6$. Now, for a given $B\in \F_{2^n}^{*}\setminus \{1\} $ with $\delta_{F_2}(B)=6$, we can write $$b^{2^m}+ Bb^2 + (B + 1)b = 0.$$ Since B is fixed, the above equation will have eight solutions $b \in \F_{2^n}$ including
$\{0, 1\}$. Thus, for each $B\in \F_{2^n}^{*}\setminus \{1\} $ with $\delta_{F_2}(B)=6$ the above equation will have 6 solutions as $b \in \F_{2^n}^{*}\setminus\{1\}$. This gives us $\# S_6=6\left(\dfrac{2^{n-2}-5}{6}-\dfrac{K(1)}{8}\right)$ when $m \equiv 1 \pmod 3$. Hence, the number of vanishing flats for $F_2(X)=X^{2^m-1}$ over $\F_{2^{2m+1}}$, when $m \equiv 1 \pmod 3$ follows from Corollary~\ref{R1} and Corollary~\ref{T51}. 

In a similar way, when $m \not \equiv 1 \pmod 3$, we get $\# S_6=6\left(\dfrac{2^{n-2}+1}{6}-\dfrac{K(1)}{8}\right)$. This gives us the number of vanishing flats for $F_2(X)=X^{2^m-1}$ over $\F_{2^{2m+1}}$, for $m \not \equiv 1 \pmod 3$.
\end{proof}

In terms of the new concept, the authors of~\cite{Bouk}  show that $F$ is an APN function of $\F_{2^n}$ if and only if the FBCT of $F$ is zero (already known from~\cite{CV95}). Hence, the next best possible value of the FBCT of $F$ over even characteristic of functions with differential uniformity greater than two is four. We now consider a power function $F_3(X)=X^{2^{\frac{n+3}{2}}-1}$, which is a permutation if and only if $n \equiv 0 \pmod 3$. Blondeau et al. in~\cite{BP} conclude that the function $F_3$ is differentially $6$-uniform over $\F_{2^n}$, when $n$ is odd. In the following theorem, we show that the function $F_3$ attains the best possible value of the FBCT, i.e. four, for odd $n$. 

\begin{cor}\label{T52} Let $F_3(X) = X^d$ be a power function of $\F_{2^n}$, where $d = 2^{\frac{n+3}{2}}-1$ and $n$ is odd. Then for $b \in \F_{2^n}$ we have:
\begin{enumerate}
 \item[$(1)$] When $n \equiv 0 \pmod 3$,
\begin{equation*}
\nabla_{F_3}(1,b)=
\begin{cases}
   \ 2^n &\text{\rm ~if~}  b=0, \text{\rm ~or~} b=1, \\
   \ 4  &\text{\rm ~if~}  b\neq 0,1,  \text{\rm ~and~} b^{2^{m}-1}=1, \\
   \ 4  &\text{\rm ~if~}  b \in \mathcal S_6, \\
   \ 0 & \text{\rm ~otherwise.}
  \end{cases}
 \end{equation*}
\item[$(2)$] When $n \not \equiv 0 \pmod 3$,
\begin{equation*} 
\nabla_{F_3}(1,b)=
  \begin{cases}
   \ 2^n &\text{\rm ~if~}  b=0, \text{\rm ~or~} b=1, \\
   \ 4  &\text{\rm ~if~}  b \in \mathcal S_6,\\
  \ 0&\text{\rm ~otherwise,}
  \end{cases}
 \end{equation*}
\end{enumerate}
where $S_6=\left\{b \in \F_{2^n}^{*}\setminus \{1\}\,\Big|\,  b^{2^{\frac{n+3}{2}}-2} \neq 1, b^{2^{\frac{n+3}{2}}-1} \neq 1\mbox{~and~} \delta_{F_3}\left ( \dfrac{b^{2^{\frac{n+3}{2}}}+b}{b(b+1)}\right)=6\right\}$. Thus, the Feistel boomerang uniformity of $F_3$ is $4$.
\end{cor}
\begin{proof}
Clearly, $t=\frac{n+3}{2}$ and $\gcd(t,n)=3$ if $n\equiv 0 \pmod 3$ and $1$ if $n \not \equiv 0 \pmod 3$, respectively. Moreover, $\gcd(t-1,n)=1$ for all positive integers $n$. Let $B= \dfrac{b^{2^t}+b}{b(b+1)}$ when $b\neq0,1$. Then, from Theorem 9 and Corollary 2 in ~\cite{BCC}, we have $\delta_{F_3}(0)=6$ when $n \equiv 0 \pmod 3$ and $\delta_{F_3}(0)=0$ when $n \not \equiv 0 \pmod 3$. Also, we have $\delta_{F_3}(1)=2$ for every positive integer $n$ and $\delta_{F_3}(B)=6=2^3-2$ for some $B \in \F_{2^n}^{*}\setminus \{1\}$. Therefore, $\max(\delta_{F_3}(B)-2,0)=4$ when $B \neq 0,1$. Thus, from Theorem~\ref{thmt}, we get the following:
\begin{enumerate}
 \item If $n \equiv 0 \pmod 3$,
 $$\nabla_{F_3}(1,b)=\left\{\begin{array}{ll}
2^n&\text{ if }b=0,1,\\
\delta_{F_3}(B)&\text{ if }B=1,\\
\max(\delta_{F_3}(B)-2,0) &\text{ if }B\neq1,\\
0&\text{ \rm otherwise.~}
\end{array}\right.=\left\{\begin{array}{ll}
2^n&\text{ if }b=0,1,\\
0 &\text{ if }B=1,\\
4 &\text{ if }B=0,\\
4 &\text{ if }B\neq0,1,\\
0&\text{ \rm otherwise.~}
\end{array}\right.$$
\item If $n \not \equiv 0 \pmod 3$,
$$\nabla_{F_3}(1,b)=\left\{\begin{array}{ll}
2^n &\text{ if }b=0,1,\\
\delta_{F_3}(B) &\text{ if }B=1,\\
\max(\delta_{F_3}(B)-2,0)&\text{ if }B\neq1,\\
0&\text{ \rm otherwise.~}
\end{array}\right.=\left\{\begin{array}{ll}
2^n&\text{ if }b=0,1,\\
0&\text{ if }B=1,\\
0&\text{ if }B=0,\\
4 &\text{ if }B\neq0,1,\\
0&\text{ \rm otherwise.}
\end{array}\right.$$
\end{enumerate}
It is obvious that $B=1$ if and only if $b^{2^t-2}=1, b\neq1$. Since, $\gcd(2^m-2,2^n-1)=1$, there does not exist $b \in \F_{2^n}^{*} \setminus \{1\}$ such that $B=1$. Moreover, we have $B=0$ if and only if $b^{2^t-1}=1,b\neq1.$ Since, $\gcd(2^t-1,2^n-1)=1$ when $n \not \equiv 0 \pmod 3$, so in that case $B=0$ is not possible. Thus, we are left with the case when $n \equiv 0 \pmod 3$; in this case $\gcd(2^t-1,2^n-1)=7$, giving us exactly six values of $b \in \F_{2^n}^{*} \setminus \{1\}$ such that $B=0$.

Next, we have $\nabla_{F_3}(1,b)=4$, when $\delta_{F_3}(B)=6, B\neq0,1$, or equivalently, when $b \in \mathcal S_6 =\left\{b \in \F_{2^n}^{*}\setminus \{1\},  b^{2^t-2} \neq 1, b^{2^t-1} \neq 1 \mbox{~and~} \delta_{F_3}\left ( \dfrac{b^{2^t}+b}{b(b+1)}\right)=6\right\}$.
\end{proof}

We next determine the number of vanishing flats of $F_3$.
\begin{cor}
Let  $F_3(X)=X^{2^t-1}$ over $\F_{2^{n}}$, where $t=\frac{n+3}{2}$. Then we have
$$
\# \VB_{n,F_3}=\dfrac{2^{n-2}+1}{6}-\dfrac{K(1)}{8},
$$
where $K(1)$ is the Kloosterman sum.
\end{cor}
\begin{proof} 
It is known from Corollary 2 in \cite{BCC} that $\#S_6$ for $F_3(X)=X^{2^{\frac{n+3}{2}}-1}$ is same as $\#S_6$ for $F_{2}(X)=X^{2^m-1}$. Thus, from Corollary~\ref{R2} we get, $\# S_6=6\left(\dfrac{2^{n-2}-5}{6}-\dfrac{K(1)}{8}\right)$ when $n \equiv 0 \pmod 3$ and $\# S_6=6\left(\dfrac{2^{n-2}+1}{6}-\dfrac{K(1)}{8}\right)$ when $n \not \equiv 0 \pmod 3$, respectively. Hence, using Corollary~\ref{R1} and Corollary~\ref{T52}, we get the number of vanishing flats for $F_3(X)=X^{2^t-1}$ over $\F_{2^{n}}$.
\end{proof}

\section{The second-order zero differential spectrum for some differentially low uniform permutation
polynomials in even characteristic}
\label{S5}
Next, we shall consider two differentially low uniform permutation polynomials over the finite field $\F_{2^n}$. The first one was introduced by Tan et al.~\cite{TQTL} who showed that when $n$ is even, the permutation polynomial $F(X) = X^{-1}+\Trn\left(\frac{X^2}{X+1}\right)$ is differentially $4$-uniform. Further, Hasan et al. in~\cite{HPS} studied the $c$-differential and boomerang uniformities of $F(X)$. 

Below, we shall be using the following result of~\cite{Leonard-Williams}.
Let $f$ be the polynomial $f(X)=X^{4}+a_{2}X^{2}+a_{1}X+a_{0}$ with $a_{0}a_{1} \neq 0$ and $g(Y)=Y^{3}+a_{2}Y+a_{1}$ be the companion cubic with the roots $r_{1},\,r_{2},\,r_{3}$. When the roots exist in $\mathbb{F}_{2^{n}}$, we set $w_{i}=a_{0}r_{i}^{2}/a_{1}^{2}$. We write a polynomial $h$ as $h=(1,2,3,\ldots)$ over some field to mean that it decomposes as a product of degree $1,2,3,\ldots$, over that field.
\begin{lem}[\textup{\cite{Leonard-Williams}}]
\label{quartic_equation}
Let $f(X)=X^{4}+a_{2}X^{2}+a_{1}X+a_{0}\in \mathbb{F}_{2^{n}}[X]$ with $a_{0}a_{1} \neq 0$. The factorization of $f(X)$ over $\mathbb{F}_{2^{n}}$ is characterized as
follows:
\begin{itemize}
\item[$(i)$] $f=(1,1,1,1) \Leftrightarrow g=(1,1,1)$ and $\Trn\left(w_{1}\right)=\Trn\left(w_{2}\right)=\Trn\left(w_{3}\right)=0$;

\item[$(ii)$] $f=(2,2) \Leftrightarrow g=(1,1,1)$ and $\Trn\left(w_{1}\right)=0, \Trn\left(w_{2}\right)=\Trn\left(w_{3}\right)=1$;

\item[$(iii)$] $f=(1,3) \Leftrightarrow g=(3)$;

\item[$(iv)$] $f=(1,1,2) \Leftrightarrow g=(1,2)$ and $\Trn\left(w_{1}\right)=0$;

\item[$(v)$] $f=(4) \Leftrightarrow g=(1,2)$ and $\Trn\left(w_{1}\right)=1$.
\end{itemize}
\end{lem}

In the following theorem, we give the FBCT of $F(X) = X^{-1}+\Trn\left(\frac{X^2}{X+1}\right)$.

\begin{thm}\label{T6}
 Let $F(X) = X^{-1}+\Trn\left(\dfrac{X^2}{X+1}\right)$ be a function over $\F_{2^n}$, where $n$ is even. Then for $a,b\in \F_{2^n}$,
\begin{equation*} \nabla_F(a,b)=
\begin{cases}
4 &~\mbox{~if~} \Trn(b^{-1})=\Trn(b^{-1}\omega)=\Trn(b^{-1}\omega^2)=0, \Trn\left(\frac{b^3}{b^3+1}\right)=\Trn(b^3)=1,\\
&~\mbox{~or~} \Trn\left(\frac{b^3}{b^3+1}\right)=\Trn(b^3)=0 \mbox{~and~} \Trn(b^{-1})=1,\\
&~\mbox{~or~} \Trn\left(\frac{b^3}{b^3+1}\right)=\Trn(b^3)=0 \mbox{~and~} \Trn(b^{-1}\omega)=1,\\
&~\mbox{~or~} \Trn\left(\frac{b^3}{b^3+1}\right)=\Trn(b^3)=0 \mbox{~and~} \Trn(b^{-1}\omega^2)=1,\\
&~\mbox{~or~} \Trn(w_1)=\Trn(w_2)=\Trn(w_3)=\Trn\left(\frac{ab(a+b)}{a^2+b^2+ab(a+b)+ab+1}\right)=1,\\
8 &~\mbox{~if~} \Trn\left(\frac{b^3}{b^3+1}\right)=\Trn(b^{-1})=\Trn(b^{-1}\omega)=\Trn(b^{-1}\omega^2)=0, \Trn(b^3)=1,\\
2^n &~\mbox{~if~} ab=0 \mbox{~or~} a=b,\\
0 &~\mbox{~otherwise,~}\\
\end{cases}
\end{equation*}
where $\omega$ is a cube root of unity, $w_1=\frac{a}{b(a+b)}, w_2=\frac{b}{a(a+b)}$ and $w_3=\frac{a+b}{ab}$. 
Moreover, $F$ is second-order zero differential $8$-uniform (that is, the Feistel boomerang uniformity of $F$ is~$8$).
\end{thm}
\begin{proof}
 For $a,b \in \F_{2^n}$, we consider the equation: 
\begin{equation}\label{eqF}
F(X + a + b) +F (X + b) +F(X + a) + F (X) = 0.
\end{equation}
If $ab = 0$ and $a=b$, then $\nabla_F(a, b) = 2^n$.
If $ab \neq 0$, $a\neq b$, then we can write it as,
\begin{align*}
& (X + a + b)^{-1} +(X + b)^{-1} +(X + a)^{-1} + X^{-1} \\ 
& \qquad \qquad + \Trn\left(\dfrac{X^2}{X+1} + \dfrac{X^2+a^2}{X+a+1} +\dfrac{X^2+b^2}{X+b+1} +\dfrac{X^2+a^2+b^2}{X+a+b+1} \right)=0.
\end{align*}

\textbf{Case A.} $a^2+ab+b^2=0$ (i.e. $a\in\{b\omega,b(\omega+1)\})$. We can distinguish two subcases.

\textbf{Subcase a.} Let $\Trn\left(\dfrac{X^2}{X+1} + \dfrac{X^2+a^2}{X+a+1} +\dfrac{X^2+b^2}{X+b+1} +\dfrac{X^2+a^2+b^2}{X+a+b+1} \right)=0.$ Then 
$\nabla_F(a,b)=\nabla_G(a,b)$, where $G(X)=X^{-1}$, and therefore, from Lemma~\ref{L1},  there are a priori 4 solutions, namely $X=0,a,b$ and $a+b$.  We have that, for all these, 
$$\Trn\left(\dfrac{X^2}{X+1} + \dfrac{X^2+a^2}{X+a+1} +\dfrac{X^2+b^2}{X+b+1} +\dfrac{X^2+a^2+b^2}{X+a+b+1} \right)=\Trn\left( \dfrac{b^3}{b^3+1} \right).$$

Therefore, in this subcase, there are four solutions, namely $X=0,a,b$ and $a+b$, if  $\Trn\left( \dfrac{b^3}{b^3+1} \right)=0$, and no solutions  if  $\Trn\left( \dfrac{b^3}{b^3+1} \right)=1$.

\textbf{Subcase b.} Let $\Trn\left(\dfrac{X^2}{X+1} + \dfrac{X^2+a^2}{X+a+1} +\dfrac{X^2+b^2}{X+b+1} +\dfrac{X^2+a^2+b^2}{X+a+b+1} \right)=1.$ Then 
\begin{equation*} 
 (X + a + b)^{-1} +(X + b)^{-1} +(X + a)^{-1} + X^{-1}=1.
\end{equation*}

If $X=0,a,b$ or $a+b$, this renders
$$\frac{1}{a}+\frac{1}{b}+\frac{1}{a+b}=1,$$
or, equivalently,
$$a^2+ab+b^2=ab(a+b).$$
However, $a^2+ab+b^2=0$, gives a contradiction. Assume, then, $X\not \in \{0,a,b,a+b\}$. In this case, we obtain $$\dfrac{1}{X}+\dfrac{1}{X+a}+\dfrac{1}{X+b}+\dfrac{1}{X+a+b}=1,$$
which is the same as,
\begin{equation*}
 X^4+(a^2+b^2+ab)X^2+ab(a+b)X+ab(a+b)=0.
\end{equation*}
Since $a\in\{b\omega,b(\omega+1)\}$, this gives 
\begin{equation}\label{ee1}
 X^4+b^3X+b^3=0.
\end{equation}
Looking at the roots of the resolvent cubic $Y^3+b^3$, we see that there are three different roots $Y_1=b,\,Y_2=b\omega,\,Y_3=b(\omega+1)$. Equation \eqref{ee1} has then four solutions (different than $0,a,b,a+b$) if $\Trn(b^{-1})=\Trn(b^{-1}\omega)=\Trn(b^{-1}\omega^2)=0$, and no solutions if $\Trn(b^{-1})=1$ or $\Trn(b^{-1}\omega)=1$ or $\Trn(b^{-1}\omega^2)=1$. The solutions, if any, have to fullfill the condition $\Trn\left(\dfrac{X^2}{X+1} + \dfrac{X^2+a^2}{X+a+1} +\dfrac{X^2+b^2}{X+b+1} +\dfrac{X^2+a^2+b^2}{X+a+b+1} \right)=\Trn(b^3)=1.$

Summarizing these two subcases, then, we see that Equation~\eqref{eqF} has 8 solutions if $\Trn\left( \dfrac{b^3}{b^3+1} \right)=\Trn(b^{-1})=\Trn(b^{-1}\omega)=\Trn(b^{-1}\omega^2)=0$ and $\Trn(b^3)=1$, 4 solutions if $\Trn(b^{-1})=\Trn(b^{-1}\omega)=\Trn(b^{-1}\omega^2)=0$ and $\Trn\left( \dfrac{b^3}{b^3+1} \right)=\Trn(b^3)=1$ or $\Trn\left( \dfrac{b^3}{b^3+1} \right)=0$, $\Trn(b^{-1})=1$ or $\Trn(b^{-1}\omega)=1$ or $\Trn(b^{-1}\omega^2)=1$ or $\Trn(b^3)=0$, and no solutions otherwise.

\textbf{Case B.} $a^2+ab+b^2\neq0$ (i.e. $a\notin\{b\omega,b(\omega+1)\})$. We can distinguish two subcases:

\textbf{Subcase a.} Let $\Trn\left(\dfrac{X^2}{X+1} + \dfrac{X^2+a^2}{X+a+1} +\dfrac{X^2+b^2}{X+b+1} +\dfrac{X^2+a^2+b^2}{X+a+b+1} \right)=0.$ Then 
$\nabla_F(a,b)=\nabla_G(a,b)$, where $G(X)=X^{-1}$, and therefore, from Lemma~\ref{L1},  there are no solutions.

\textbf{Subcase b.} Let $\Trn\left(\dfrac{X^2}{X+1} + \dfrac{X^2+a^2}{X+a+1} +\dfrac{X^2+b^2}{X+b+1} +\dfrac{X^2+a^2+b^2}{X+a+b+1} \right)=1.$ Then 
\begin{equation}\label{eee}
 (X + a + b)^{-1} +(X + b)^{-1} +(X + a)^{-1} + X^{-1}=1.
\end{equation}
If $X=0,a,b$ or $a+b$, this renders
$$\frac{1}{a}+\frac{1}{b}+\frac{1}{a+b}=1,$$
or, equivalently,
$$a^2+ab+b^2=ab(a+b).$$

This is possible, and the solutions are valid if
\allowdisplaybreaks
\begin{align*}
  & \Trn\left(\dfrac{X^2}{X+1} + \dfrac{X^2+a^2}{X+a+1} +\dfrac{X^2+b^2}{X+b+1} +\dfrac{X^2+a^2+b^2}{X+a+b+1} \right)\\ & \qquad \qquad \qquad \qquad =\Trn\left(\frac{ab(a+b)}{a^2+b^2+ab(a+b)+ab+1}\right) =1.
 \end{align*}
Because we have $a^2+ab+b^2+ab(a+b)=0$, the above condition will further give us $\Trn(ab(a+b))=1$. We would next show that there does not exist any $(a,b)  \in \F_{2^n} \times \F_{2^n}$ satisying $a^2+b^2+ab+ab(a+b)=0$ and $\Trn(ab(a+b))=1$. Let us assume that $a^2+b^2+ab+ab(a+b)=0$ for some $(a,b) \in \F_{2^n} \times \F_{2^n}$. Then applying the absolute trace on the above equation, we get that
 \begin{equation}\label{t1}
  \Trn(ab(a+b))=\Trn(a+b+ab).
 \end{equation}
Now dividing $a^2+b^2+ab+ab(a+b)=0$ by $a$ and applying the trace again, we get that
\begin{equation*}
  \Trn(a)+\Trn\left(\frac{b^2}{a}\right)+\Trn(b)+\Trn(ba)+\Trn(b^2)=0,
 \end{equation*}
 or equivalently,
 \begin{equation}\label{t2}
  \Trn(a)+\Trn\left(\frac{b^2}{a}\right)+\Trn(ba)=0.
 \end{equation}
 Again dividing $a^2+b^2+ab+ab(a+b)=0$ by $a^2$ and applying the absolute trace, we have
 \begin{equation*}
  \Trn(1)+\Trn\left(\frac{b^2}{a^2}\right)+ \Trn\left(\frac{b}{a}\right)+\Trn(b)+\Trn\left(\frac{b^2}{a}\right)=0,
 \end{equation*}
 which is the same as
 \begin{equation}\label{t3}
 \Trn(b)+\Trn\left(\frac{b^2}{a}\right)=0.
 \end{equation}
 From Equation~\eqref{t1}, Equation~\eqref{t2} and Equation~\eqref{t3}, we get that for even $n$, there does not exist any $(a,b)$ in $\F_{2^n} \times \F_{2^n}$ such that  $\Trn(ab(a+b))=1$. Hence $0,a,b$ and $a+b$ are not the solutions of Equation~\eqref{eee}.
 
If $X\not \in \{0,a,b,a+b\}$, we obtain $$\dfrac{1}{X}+\dfrac{1}{X+a}+\dfrac{1}{X+b}+\dfrac{1}{X+a+b}=1,$$
which is same as, 
\begin{equation}\label{ee3}
 X^4+(a^2+b^2+ab)X^2+ab(a+b)X+ab(a+b)=0.
\end{equation}
Again, via~Lemma~\ref{quartic_equation}, observe that the companion cubic polynomial is
yet again 
\[
Y^3+(a^2 + a b + b^2) Y +a b (a + b) = (Y+a)(Y+b)(Y+a+b).
\]
We let  $a_0=ab(a+b), a_1=a b (a+b), a_2=a^2+a b+b^2$, $w_1=\frac{a_0a^2}{a_1^2}=\frac{a}{b(a+b)}, w_2=\frac{a_0 b^2}{a_1^2}=\frac{b}{a(a+b)},$ and $w_3=\frac{a_0 (a+b)^2}{a_1^2}=\frac{a+b}{ab}$. Since the companion cubic  is thus $(1,1,1)$, from  Lemma~\ref{quartic_equation} $(i)$ and $(ii)$, we then see  that the polynomial~\eqref{ee3} is either $(1,1,1,1)$ or $(2,2)$, and so, it has a root in $\F_{2^n}$  if and only if all the traces of $w_1,w_2,w_3$ are zero. Thus, Equation~\eqref{ee3} can have at most four solutions. Hence,  Equation~\eqref{eqF} can have at most four solutions from Subcase b.

In conclusion, Equation~\eqref{eqF} has at most 8 solutions.
\end{proof}

Next, for $t$ and $n$, two positive integers, the authors in~\cite{ZZ} show that the function $F(X)= \dfrac{1}{X+\gamma \Trn(X^{2^t+1})}$ is a permutation of $\F_{2^n}$ with $\delta_F \leq 6$, where $\gamma$ is an element of $\F_{2^n} \cap \F_{2^{2t}}^{*}$ with $\Trn(\gamma^{2^t+1})=0.$ In the following theorem, we determine the FBCT of this function. We have not explicitly stated the second-order zero differential spectrum of this function in the statement of the theorem due to the large number of conditions required. However, it is readily apparent from the proof of the theorem.
 \begin{thm}\label{T7}
Let $F(X)= \dfrac{1}{X+\gamma \Trn(X^{2^t+1})}$ where $t$ and $n$ are two positive integers with $t<n$, and $\gamma$ is an element of $\F_{2^n} \cap \F_{2^{2t}}^{*}$ with $\Trn(\gamma^{2t+1})=0.$ Then for $a,b\in \F_{2^n}$, $\nabla_F(a,b)\in \{0,4,8\}$. Consequently, $F$ has Feistel boomerang uniformity at most $8$.
\end{thm}
\begin{proof}
For $a,b \in \F_{2^n}$, we consider the equation: 
\begin{equation}\label{ll1}
F(X + a + b) +F (X + b) +F(X + a) + F (X) = 0.
\end{equation}
If $ab = 0$ or $a=b$, then $\nabla_F(a, b) = 2^n$.
If $ab \neq 0$,  Equation~\eqref{ll1} can be written as,
\begin{align*}
& \dfrac{1}{X+\gamma \Trn(X^{2^t+1})} + \dfrac{1}{X+a+\gamma \Trn((X+a)^{2^t+1})} \\
& \qquad + \dfrac{1}{X+b+\gamma \Trn((X+b)^{2^t+1})}+ \dfrac{1}{X+a+b+\gamma
\Trn((X+a+b)^{2^t+1})}=0.
\end{align*}
Next, we consider sixteen possible cases as pointed out in Table~\ref{Table3}.
\allowdisplaybreaks
\begin{table}
\caption{Various cases when $ab \neq 0$ and $a \neq b$}
\label{Table3} 
\begin{center}
\scalebox{0.72}{
\begin{tabular}{|c|c|c|c|c|c|c|c|} 
\hline
Case&$\Trn(X^{2^t+1})$&$\Trn((X+a)^{2^t+1})$&$\Trn((X+b)^{2^t+1})$&$\Trn((X+a+b)^{2^t+1})$& Equation \\
\hline 
$1$&$0$ & $0$ & $0$ & $0$ &  $\frac{1}{X}+\frac{1}{X+a}+\frac{1}{X+b}+\frac{1}{X+a+b}=0$  \\
\hline 
$2$&$1$ & $0$ & $0$ & $0$ &   $\frac{1}{X+\gamma}+\frac{1}{X+a}+\frac{1}{X+b}+\frac{1}{X+a+b}=0$  \\
  \hline 
  $3$&$0$ & $1$ & $0$ & $0$ &  $\frac{1}{X}+\frac{1}{X+a+\gamma}+\frac{1}{X+b}+\frac{1}{X+a+b}=0$   \\
  \hline 
  $4$&$0$ & $0$ & $1$ & $0$ &  $\frac{1}{X}+\frac{1}{X+a}+\frac{1}{X+b+\gamma}+\frac{1}{X+a+b}=0$   \\
  \hline 
  $5$&$0$ & $0$ & $0$ & $1$ &   $\frac{1}{X}+\frac{1}{X+a}+\frac{1}{X+b}+\frac{1}{X+a+b+\gamma}=0$  \\
  \hline 
  $6$&$1$ & $1$ & $0$ & $0$ &   $\frac{1}{X+\gamma}+\frac{1}{X+a+\gamma}+\frac{1}{X+b}+\frac{1}{X+a+b}=0$  \\
  \hline 
  $7$&$1$ & $0$ & $1$ & $0$ &   $\frac{1}{X+\gamma}+\frac{1}{X+a}+\frac{1}{X+b+\gamma}+\frac{1}{X+a+b}=0$  \\
  \hline 
  $8$&$1$ & $0$ & $0$ & $1$ &  $\frac{1}{X+\gamma}+\frac{1}{X+a}+\frac{1}{X+b}+\frac{1}{X+a+b+\gamma}=0$   \\
  \hline 
  $9$&$0$ & $1$ & $1$ & $0$ & $\frac{1}{X}+\frac{1}{X+a+\gamma}+\frac{1}{X+b+\gamma}+\frac{1}{X+a+b}=0$    \\
  \hline 
  $10$&$0$ & $1$ & $0$ & $1$ &   $\frac{1}{X}+\frac{1}{X+a+\gamma}+\frac{1}{X+b}+\frac{1}{X+a+b+\gamma}=0$  \\
  \hline 
  $11$&$0$ & $0$ & $1$ & $1$ &  $\frac{1}{X}+\frac{1}{X+a}+\frac{1}{X+b+\gamma}+\frac{1}{X+a+b+\gamma}=0$   \\
  \hline 
  $12$&$1$ & $1$ & $1$ & $0$ &  $\frac{1}{X+\gamma}+\frac{1}{X+a+\gamma}+\frac{1}{X+b+\gamma}+\frac{1}{X+a+b}=0$   \\
  \hline 
  $13$&$1$ & $0$ & $1$ & $1$ &  $\frac{1}{X+\gamma}+\frac{1}{X+a}+\frac{1}{X+b+\gamma}+\frac{1}{X+a+b+\gamma}=0$  \\
  \hline 
  $14$&$0$ & $1$ & $1$ & $1$ &  $\frac{1}{X}+\frac{1}{X+a+\gamma}+\frac{1}{X+b+\gamma}+\frac{1}{X+a+b+\gamma}=0$  \\
  \hline 
  $15$&$1$ & $1$ & $0$ & $1$ &   $\frac{1}{X+\gamma}+\frac{1}{X+a+\gamma}+\frac{1}{X+b}+\frac{1}{X+a+b+\gamma}=0$ \\
  \hline 
  $16$&$1$ & $1$ & $1$ & $1$ &  $\frac{1}{X+\gamma}+\frac{1}{X+a+\gamma}+\frac{1}{X+b+\gamma}+\frac{1}{X+a+b+\gamma}=0$  \\
  \hline 
 \end{tabular}
 }
 \end{center}
 \end{table}
 Let $N_i$ denote the number of solutions in $\F_{2^n}$ for the $i^{th}$ case given in Table~\ref{Table3}.
 
 \textbf{Case 1.} For the first case, we need to consider the number of solutions $X \in \F_{2^n}$ of the equation,
 $\dfrac{1}{X}+\dfrac{1}{X+a}+\dfrac{1}{X+b}+\dfrac{1}{X+a+b}=0$, which is the same as the FBCT of the inverse function. From Lemma~\ref{L2}, it is clear that for odd $n$, Case 1 has no solution. When $n$ is even, then Lemma~\ref{L1} would give us only four solutions $X \in \{0,a,b,a+b\}$ of Case 1 if and only if $a\in\{b\omega,b\omega^2\}$ and $\Trn(a^{2^t+1})=\Trn(b^{2^t+1})=\Trn((a+b)^{2^t+1})=0$.
 
\textbf{Case 2.} In this case, consider the number of solutions $X \in \F_{2^n}$ of the equation,
 $\dfrac{1}{X+\gamma}+\dfrac{1}{X+a}+\dfrac{1}{X+b}+\dfrac{1}{X+a+b}=0$. It is clear that $X= \gamma$ cannot be a solution of Case 2 as $\Trn(\gamma^{2^t+1})=0$. Moreover, $X=a,b$ and $a+b$ cannot satisfy Case 2, simultaneously. This is because $X=a$ is solution if $a^2+b^2+ab+a\gamma=0$ along with $\Trn(a^{2^t+1})=1$ and $\Trn(b^{2^t+1})=0=\Trn((a+b)^{2^t+1})$. Similarly, $X=b$ is solution if $a^2+b^2+ab+b\gamma=0$ and $\Trn(b^{2^t+1})=1$ and $\Trn(a^{2^t+1})=0=\Trn((a+b)^{2^t+1})$. Finally, $X=a+b$ is solution if $a^2+b^2+ab+(a+b)\gamma=0$ and $\Trn((a+b)^{2^t+1})=1$ and $\Trn(a^{2^t+1})=0=\Trn((b^{2^t+1})$. Summarizing the above, we get at most one solution of Case 2 when $X \in \{a,b,a+b\}$. For the other values of $X$, one can get a unique solution $X' \in \F_{2^n}$ satisfying $X'^2=a^2+b^2+ab+\dfrac{ab(a+b)}{\gamma}$, $\Trn(X'^{2^t+1})=1$ and $\Trn((X'+a)^{2^t+1})=\Trn((X'+b)^{2^t+1})=\Trn((X'+a+b)^{2^t+1})=0$.
 
 \textbf{Case 3.} Here, $N_2=N_3$, because if $X$ is solution for Case~2, then $X+a$ will satisfy Case~3 and hence Case~3 has at most two solutions. 
 
 \textbf{Case 4.} Similar to Case 3, we have $N_2=N_4$, because if $X$ is a solution for Case~2, then $X+b$ will satisfy Case 4 and hence Case~4 has at most two solutions.

 \textbf{Case 5.} Similar to  Case 2, we have $N_2=N_5$.
 
 Summarizing Cases 2--5, we have 
 \allowdisplaybreaks
 \begin{equation}\label{n3} X \in
 \begin{cases}
 \{0,a,b,a+b\} &~\mbox{~if~} a^2+b^2+ab+a\gamma=0, \Trn(a^{2^t+1})=1 \mbox{~and~} \\ & \quad  \Trn(b^{2^t+1})=0=\Trn((a+b)^{2^t+1}),\\
  &~\mbox{~or, if~} a^2+b^2+ab+b\gamma=0, \Trn(b^{2^t+1})=1 \mbox{~and~}\\ & \quad  \Trn(a^{2^t+1})=0=\Trn((a+b)^{2^t+1}), \mbox{~or,~} \\
 &~a^2+b^2+ab+(a+b)\gamma=0, \Trn((a+b)^{2^t+1})=1 \\ & \mbox{~and~}   \Trn(b^{2^t+1})=0=\Trn(a^{2^t+1}),\\
 \{X',X'+a,X'+b,X'+a+b\} &~\mbox{~if~} \Trn(X'^{2^t+1})=1 \mbox{~and~} \Trn((X'+a)^{2^t+1})=0 \\ & \quad  \Trn((X'+b)^{2^t+1})=0=\Trn((X'+a+b)^{2^t+1}),\\
 \end{cases}
 \end{equation}
 where $X'^2=a^2+b^2+ab+\dfrac{ab(a+b)}{\gamma}$ and $\gamma \not \in \{a,b,a+b\}$.

 \textbf{Case 6.} In this case, we have:
 \begin{equation}\label{eb2}
 \frac{1}{X+\gamma}+\frac{1}{X+a+\gamma}+\frac{1}{X+b}+\frac{1}{X+a+b}=0,
 \end{equation}
 where $\Trn(X^{2^t+1})=\Trn((X+a)^{2^t+1})=1$ and $\Trn((X+b)^{2^t+1})=\Trn((X+a+b)^{2^t+1})=0.$ If $\gamma=b$, then $\Trn(b^{2^t+1})=0$, and this will give us $0=\Trn((X+b)^{2^t+1})=\Trn(X^{2^t+1})=1$. Hence, there is no solution corresponding to $\gamma=b$. Similarly, if $\gamma=a+b$, then $\Trn((a+b)^{2^t+1})=0$, but this will give us a contradiction as $0=\Trn((X+a+b)^{2^t+1})=\Trn(X^{2^t+1}))=1$. 
 
 Now, let $\gamma \not \in \{b,a+b\}$,  then $X \in \{\gamma,a+\gamma\}$ does not satisfy Equation~\eqref{eb2} and $X \in \{b,a+b\}$ satisfy Equation~\eqref{eb2} if $\dfrac{1}{a}+\dfrac{1}{b+\gamma}+\dfrac{1}{a+b+\gamma}=0$, which is the same as $a^2+b^2+\gamma^2+a(b+\gamma)=0$. Substituting $b+\gamma = Z$, we can rewrite the previous equation as $Z^2+aZ+a^2=0$, which has no solution when $n$ is odd as $\Trn(1) \neq 0$. For even $n$, $a^2+b^2+\gamma^2+a(b+\gamma)=0$ holds for the pairs $\{(a,\gamma + a\omega), (a, \gamma+a\omega^2)\}$ in $\F_{2^n} \times \F_{2^n}$. Thus, $X \in \{ b,a+b\}$ satisfy Equation~\eqref{eb2} for the pairs $\{(a,\gamma + a\omega), (a, \gamma+a\omega^2)\}$, where $a \in \F_{2^n}$, $n$ is even and also $\Trn(b^{2^t+1})=\Trn((a+b)^{2^t+1})=1$ and $\Trn(a^{2^t+1})=0.$

 \textbf{Case 7.} Due to symmetry in $a$ and $b$, we have $N_7=N_6$, and hence it can have at most two solutions $\{a,a+b\}$, when $n$ is even and $a^2+b^2+\gamma^2+b(a+\gamma)=0$ for the pairs $(a,b)=\{(\gamma + b\omega,b), (\gamma+b\omega^2,b)\}$ in $\F_{2^n} \times \F_{2^n}$ satisfying $\Trn(a^{2^t+1})=1=\Trn((a+b)^{2^t+1}) \mbox{~and~} \Trn(b^{2^t+1})=0$. 
 
 \textbf{Case 8.} Here, we have
 \[
 \frac{1}{X+\gamma}+\frac{1}{X+a}+\frac{1}{X+b}+\frac{1}{X+a+b+\gamma}=0,
 \]
 where $\Trn(X^{2^t+1})=\Trn((X+a+b)^{2^t+1})=1$ and $\Trn((X+b)^{2^t+1})=\Trn((X+a)^{2^t+1})=0.$ If $\gamma \in \{a,b\}$, then there does not exist any solution $X \in \F_{2^n}$. Let $\gamma \not \in \{a,b\}$, then clearly $X=\gamma$ and $X=a+b+\gamma$ does not satisfy the above equation. Clearly, $X=a$ and $X=b$ are solutions only when $a^2+b^2+\gamma^2+ab+(a+b)\gamma=0$, $\Trn(a^{2^t+1})=\Trn(b^{2^t+1})=1$ and $\Trn((a+b)^{2^t+1})=0$. Otherwise, there does not exist any solution in this case.

\textbf{Case 9.} By the same arguments as in   Case 8, we have $X=0$ and $X=a+b$ as solutions, when $a^2+b^2+\gamma^2+ab+(a+b)\gamma=0$, $\Trn(a^{2^t+1})=\Trn(b^{2^t+1})=1$ and $\Trn((a+b)^{2^t+1})=0$.

 \textbf{Case 10.} Using a similar approach as in Case 7, we have at most two solutions $\{0,b\}$, only when $n$ is even and $a^2+b^2+\gamma^2+b(a+\gamma)=0$ for the pairs $(a,b)=\{(\gamma + b\omega,b), (\gamma+b\omega^2,b)\}$ in $\F_{2^n} \times \F_{2^n}$ satisfying the required trace conditions in this case.

 \textbf{Case 11.} Following similar reasoning as in Case 6, there are at most two solutions, $\{0,a\}$, which occur only when $n$ is even and $a^2+b^2+\gamma^2+a(b+\gamma)=0$. These solutions correspond to the pairs $(a,b)=\{(\gamma + b\omega,b), (\gamma+b\omega^2,b)\}$ in $\F_{2^n} \times \F_{2^n}$ that satisfy the trace conditions required for this case.
 
 Summarizing Case 6, Case 7, Case 8, Case 9, Case 10 and Case 11, we get $X \in \{0,a,b,a+b\}$ either if $\gamma \not \in \{b,a+b\}, (a,b) \in \{(a,\gamma + a\omega), (a, \gamma+a\omega^2)\}, \Trn(b^{2^t+1})=1=\Trn((a+b)^{2^t+1})$ and $\Trn(a^{2^t+1})=0$, or if $\gamma \not \in \{a,a+b\}, (a,b) \in \{(\gamma + b\omega,b), (\gamma+b\omega^2,b)\}, \Trn(a^{2^t+1})=1=\Trn((a+b)^{2^t+1})$ and $\Trn(b^{2^t+1})=0$, or if  $\gamma \not \in \{a,b\}, a^2+b^2+\gamma^2+ab+(a+b)\gamma=0, \Trn(a^{2^t+1})=1=\Trn(b^{2^t+1})$ and $\Trn((a+b)^{2^t+1})=0.$

 \textbf{Case 12.} When $\gamma \in \{a,b,a+b\}$, Case 12 has no solutions $X \in \F_{2^n}$. Hence, assume that $\gamma \not \in \{a,b,a+b\}$. Then $X \in \{\gamma,a+\gamma,b+\gamma\}$ are not the solutions of Case 12 as $\Trn(\gamma^{2^t+1})=0$, and $X=a+b$ is solution of Case 12 if $a^2+b^2+\gamma^2+ab=0$ and $\Trn(a^{2^t+1})= \Trn(b^{2^t+1})=\Trn((a+b)^{2^t+1})=1$. When $X \not \in  \{\gamma,a+\gamma,b+\gamma,a+b\}$, then Case 12 can have a unique solution $X_{12}$ in $\F_{2^n}$ such that $X_{12}^2= \gamma^2+ab+\dfrac{ab(a+b)}{\gamma}$ and the given conditions on trace are satisfied. Hence, this case has at most two solutions $a+b$ and $X_{12}$ satisfying the above discussed conditions.
 
 \textbf{Case 13.}  Consider the following equation
 \begin{equation*}
 \frac{1}{X+\gamma}+\frac{1}{X+a}+\frac{1}{X+b+\gamma}+\frac{1}{X+a+b+\gamma}=0.
 \end{equation*}
 Clearly, $X \in \{\gamma, b+\gamma,a+b+\gamma$\} are not solutions of the above equation. Also, $X=a$ is solution only when  $a^2+b^2+\gamma^2+ab=0$ and $\Trn(a^{2^t+1})= \Trn(b^{2^t+1})=\Trn((a+b)^{2^t+1})=1$. Otherwise, it can have at most one solution $X_{13}$ satisfying $X_{13}^2= b^2+\gamma^2+ab+\dfrac{ab(a+b)}{\gamma}$, $\Trn(X_{13}^{2^t+1})= \Trn((X_{13}+b)^{2^t+1})=\Trn((X_{13}+a+b)^{2^t+1})=1$ and $\Trn((X_{13}+a)^{2^t+1})=0$.
 
 \textbf{Case 14.} Similar to Case 12, $X=0$ is solution only when  $a^2+b^2+\gamma^2+ab=0$ and $\Trn(a^{2^t+1})= \Trn(b^{2^t+1})=\Trn((a+b)^{2^t+1})=1$. Otherwise, it can have at most one solution $X_{14}$ satisfying $X_{14}^2= a^2+b^2+\gamma^2+ab+\dfrac{ab(a+b)}{\gamma}$, $\Trn((X_{14}+a)^{2^t+1})= \Trn((X_{14}+b)^{2^t+1})=\Trn((X_{14}+a+b)^{2^t+1})=1$ and $\Trn(X_{14}^{2^t+1})=0$.
 
 \textbf{Case 15.}  Similar to  Case 13, we can have at most two solutions $X=b$ when  $a^2+b^2+\gamma^2+ab=0$ and $\Trn(a^{2^t+1})= \Trn(b^{2^t+1})=\Trn((a+b)^{2^t+1})=1$ and $X=X_{15}$ satisfying $X_{15}^2= a^2+\gamma^2+ab+\dfrac{ab(a+b)}{\gamma}$, $\Trn(X_{15}^{2^t+1})= \Trn((X_{15}+a)^{2^t+1})=\Trn((X_{15}+a+b)^{2^t+1})=1$ and $\Trn((X_{15}+b)^{2^t+1})=0$.
 
 Summarizing Case 12, Case 13, Case 14 and Case 15, we have the following,
 \allowdisplaybreaks
 \begin{equation*}
 X \in
 \begin{cases}
 \{0,a,b,a+b\} &~\mbox{~if~} a^2+b^2+\gamma^2+ab=0, \Trn(a^{2^t+1})=1 \mbox{~and~} \\ & \quad  \Trn(b^{2^t+1})=1=\Trn((a+b)^{2^t+1}),\\
 \{X_{12},X_{13},X_{14},X_{15}\}  &~\mbox{~if~} \Trn(X_{12}^{2^t+1})=1 = \Trn(X_{13})^{2^t+1} \\ & \quad  \Trn(X_{15})^{2^t+1}=1 \mbox{~and~} \Trn(X_{14})^{2^t+1}=0,\\
 \end{cases}
 \end{equation*}
 where $X_{12}^2=\gamma^2+ab+\dfrac{ab(a+b)}{\gamma}$ and $\gamma \not \in \{a,b,a+b\}$.
 
 \textbf{Case 16.} This case has no solution $X \in \F_{2^n}$. This is because the only solutions possible are $X \in \{\gamma,a+\gamma,b+\gamma,a+b+\gamma\}$, when $n$ is even, but this would contradict  the fact that $\Trn(\gamma^{2^t+1})=0$.  
 
 Notice that Case 2 and Case 14 cannot occur together. This is because if $X$ satisfies Case 2, then $X+\gamma$ satisfies Case 14, which would further imply that $\Trn((X+\gamma)^{2^t+1})=0$. But as $\gamma \in \F_{2^{2t}}^{*}$ and $\Trn((\gamma)^{2^t+1})=0$, this will give us $\Trn(X^{2^t+1})=0$, which is not possible, as from Case 2, we have $\Trn(X^{2^t+1})=1$. Hence Case 2 and Case 14 cannot occur together. Similarly, the pairs (Case 3, Case 13),  (Case 4, Case 15),  (Case 5, Case 12) cannot occur simultaneously. Summarizing, we conclude that we have at most 8 solutions of Equation~\eqref{ll1} in $ \F_{2^n}$. 
 \end{proof}

\section{Conclusion}\label{S6}
In this paper, we extend the work of Boukerrou et al.~\cite{Bouk} and Li et al.~\cite{LYT} by computing the second-order zero differential spectra of some APN power functions over finite fields of odd characteristic and  functions with low differential uniformity over finite fields of even and odd characteristic in order to derive additional cryptographic properties of these functions.   It is worth noting that all of these maps exhibit a low second-order zero differential uniformity.  We also connect this concept with the sum-freedom concept, which extends the vanishing affine subspaces (of dimension~$2$) concept, in even characteristic. 
We compute, via our method, the number of vanishing flats, for many of our investigated functions (in even characteristic).
Surely, we believe that it is worthwhile to look into more functions with low differential uniformity and investigate their second-order zero differential spectrum, and perhaps even extend the concept of vanishing affine subspaces, or sum-freedom to odd characteristic (though applications are not yet visible, for such an extension).

\section*{Acknowledgements}
The authors would like to thank the editor   for the prompt handling of our paper, and  they want to extend their thanks and appreciation  to the very professional referees, who  provided beneficial and constructive comments to improve our paper.


\begin{thebibliography}{99}

\bibitem{GMiMC} M. R. Albrecht, L. Grassi, L. Perrin, S. Ramacher, C. Rechberger, D. Rotaru, A. Roy, M. Schofnegger, {\it Feistel structures for MPC, and More,} In: K. Sako, S. Schneider, P. Ryan (eds) Computer Security-ESORICS 2019, LNCS 11736, Springer, Cham, 2019.

\bibitem{Biham91} E. Biham, A. Shamir, {\it Differential cryptanalysis of DES-like cryptosystems,} J. Cryptol. 4:1 (1991), 3--72.

\bibitem{BCC} C. Blondeau, A. Canteaut, P. Charpin, {\it Differential properties of $X \rightarrow X^{2^t-1}$,} IEEE Trans. Inf. Theory 57(12) (2011), 8127--8137.

\bibitem{BP} C. Blondeau, L. Perrin, {\it More differentially $6$-uniform power functions,} Des. Codes Cryptogr. 73 (2014), 487--505.

\bibitem{Bouk} H. Boukerrou, P. Huynh, V. Lallemand, B. Mandal, M. Minier, {\it On the Feistel counterpart of the boomerang connectivity table,} IACR Trans. Symmetric Cryptol. 1 (2020), 331--362.

\bibitem{BoCa} C. Boura, A. Canteaut, {\it On the boomerang uniformity of cryptographic S-boxes,} IACR Trans. Symmetric Cryptol. 3 (2018), 290--310.

\bibitem{Carlet} C. Carlet, {\it Two generalizations of almost perfect nonlinearity,} Cryptology ePrint Archive {\url{https://ia.cr/2024/841}} (2024).

\bibitem{Carlitz} L. Carlitz, {\it Kloosterman sums and finite field extensions,} Acta Arith. 16:2 (1969), 179--183.

\bibitem{CV95}
F. Chabaud, S. Vaudenay, {\em Links between differential and linear cryptanalysis}, In: A. De Santis (ed.), Adv.  Crypt -- EUROCRYPT '94, LNCS 950, Springer, 1995, pp. 356--365.

\bibitem{CHNC} S. Choi, S. Hong, J. No, H. Chung, {\it Differential spectrum of some power functions in odd prime characteristic}, Finite Fields  Appl. 21 (2013), 11--29.

\bibitem{cid} C. Cid, T. Huang, T. Peyrin, Y. Sasaki,  L. Song, {\it Boomerang connectivity table: a new cryptanalysis tool}, In: J. Nielsen, V. Rijmen (ed) Advances in Cryptology-EUROCRYPT'18, LNCS 10821, Springer, Cham, 2018, pp. 683--714.

\bibitem{Dickson06}
L. E. Dickson, {\it Criteria for the irreducibility of functions in a finite field}, Bull. Amer. Math. Soc. 13 (1906), 1--8.

\bibitem {EM} S. Eddahmani, S. Mesnager, {\it Explicit values of the DDT, the BCT, the FBCT, and the FBDT of the inverse, the gold, and the Bracken-Leander S-boxes}. Cryptogr. Commun. 14 (2022), 1301--1344.

\bibitem{GHRS} K. Garg, S.U. Hasan, C. Riera, P. St\u anic\u a, {\it The second-order zero differential spectra of some APN and other maps over finite fields,} arXiv preprint arXiv:2310.13775 (2023).

\bibitem{HPS} S. U. Hasan, M. Pal, P. St\u anic\u a, {\it The c-Differential Uniformity and Boomerang Uniformity of Two Classes of Permutation Polynomials}, IEEE Trans. Inf. Theory 68 (2022), 679--691. 

\bibitem{Hel} T. Helleseth, R. Chunming, S. Daniel, {\it New families of almost perfect nonlinear power mappings,} IEEE Trans. Inf. Theory 45(2) (1999), 475--485.

\bibitem{HS} T. Helleseth, D. Sandberg, {\it Some power mappings with low differential uniformity}, Appl. Algebra Eng. Commun. Comput. 8 (1997), 363--370.

\bibitem{KT} M. Kuroda, S. Tsujie,
{\em A generalization of APN functions for odd characteristic},
Finite Fields and Their Applications,
Volume 47, 64--84, 
2017.


\bibitem{Leonard-Williams} P. A. Leonard,  K. S. Williams, {\it Quartics over $GF(2^n)$}, Proc. Amer. Math. Soc. 36(2) (1972), 347--350.

\bibitem{LMPPRS} S. Li, W. Meidl, A. Polujan, A. Pott, C. Riera, P. St\u anic\u a, {\it Vanishing Flats: A Combinatorial Viewpoint on the Planarity of Functions and Their Applications}, IEEE Trans. Inf. Theory 66:11 (2020), 7101--7112.

\bibitem{LYT} X. Li, Q. Yue, D. Tang, {\it The second-order zero differential spectra of almost perfect nonlinear functions and the inverse function in odd characteristic}, Cryptogr. Commun. 14(3) (2022), 653--662.

\bibitem{LN} R. Lidl, H. Niederreiter, {\it Finite Fields} Cambridge University Press, (1997).

\bibitem{MLXZ} Y. Man, N. Li, Z. Xiang, X. Zeng,{\it On the second-order zero differential spectra of some power functions over finite fields,} Cryptogr. Commun. (2024) \url{https://doi.org/10.1007/s12095-024-00740-z}.

\bibitem{MMN} Y. Man, S. Mesnager, N. Li, X. Zeng, X. Tang, {\it In-depth analysis of S-boxes over binary finite fields concerning their differential and Feistel boomerang differential uniformities}, 
Discrete Math. 347(12) (2024) \url {https://doi.org/10.1016/j.disc.2024.114185}.

\bibitem{MMM} S. Mesnager, B. Mandal, M. Msahli, {\it Survey on recent trends towards generalized differential and boomerang uniformities}. Cryptogr. Commun. 14 (2022), 691--735.

\bibitem{Nyberg} K. Nyberg, {\it Differentially uniform mappings for cryptography,} In T. Helleseth (ed), Advances in Cryptology-EUROCRYPT’93, LNCS 765, Springer, Heidelberg, 1994, pp. 55--64.

\bibitem{SSA} T. Shirai, K. Shibutani, T.  Akishita, S. Moriai, T.  Iwata, {\it The 128-bit blockcipher CLEFIA (extended abstract,)} In: Biryukov A. (eds) Fast Software Encryption-FSE 2007, LNCS 4593, Springer, Berlin, Heidelberg, 2007, pp. 181--195.

\bibitem{TQTL} Y. Tan, L. Qu, C. H. Tan, C. Li, {\it New families of differentially 4-uniform permutations over $\F_{2^{2k}}$}, in Sequences and Their Applications-SETA, LNCS 7280, T. Helleseth and J. Jedwab, Eds. Heidelberg, Germany: Springer, (2012), pp. 25--39.

\bibitem{DW} D. Wagner, {\it The boomerang attack,} In: L. R. Knudsen (ed.) Fast Software Encryption-FSE 1999. LNCS 1636, Springer, Berlin, Heidelberg, 1999, pp. 156--170. 

\bibitem{YZW} H. Yan, Z. Zhou, J. Weng, J. Wen, T. Helleseth, Q. Wang, {\it Differential spectrum of Kasami power permutations over odd characteristic finite fields,} IEEE Trans. Inf. Theory 65(10) (2019), 6819--6826.

\bibitem{ZW} Z. Zha, X. Wang, {\it Almost perfect nonlinear power functions in odd characteristic}, IEEE Trans. Inf. Theory 57(7) (2011), 4826--4832.

\bibitem{ZZ} T. Ziran, X. Zeng, {\it Non-monomial permutations with differential uniformity six,} Journal of Systems Science and Complexity 31(4) (2018), 1078--1089.

\end{thebibliography}
\end{document}